\documentclass[acmsmall]{acmart}

\usepackage{comment}
\usepackage[linesnumbered,ruled,vlined,]{algorithm2e}
\usepackage{enumitem}
\usepackage[ruled,linesnumbered]{algorithm2e}
\usepackage{mathrsfs}
\usepackage{makecell}
\usepackage{subcaption}
\usepackage{tablefootnote}
\usepackage{booktabs,tabularx}
\usepackage{amsmath}
\usepackage{graphicx}
\usepackage{adjustbox}
\usepackage{makecell}
\usepackage{caption}
\usepackage{xcolor}
\usepackage{wrapfig}
\makeatletter
\renewcommand\@fnsymbol[1]{\ifcase#1\or †\or \dagger\or \ddagger\else\@arabic{#1}\fi}
\makeatother

\newcommand{\ML}{\textsc{Camal}}

\newcommand{\revise}[1]{\textcolor{black}{#1}}%
\AtBeginDocument{%
  }

\setcopyright{acmlicensed}
\copyrightyear{2024}
\acmJournal{PACMMOD}
\acmYear{2024} \acmVolume{2} \acmNumber{N4 (SIGMOD)}
\acmArticle{202} \acmMonth{9} \acmPrice{15.00}
\acmDOI{10.1145/3677138}

\begin{document}

\title{CAMAL: Optimizing LSM-trees via Active Learning}

\author{Weiping Yu}
\affiliation{%
  \institution{Nanyang Technological University}
  \country{Singapore}}
\email{WEIPING001@e.ntu.edu.sg}

\author{Siqiang Luo}
\authornote{Corresponding Author}
\affiliation{%
  \institution{Nanyang Technological University}
  \country{Singapore}}
\email{siqiang.luo@ntu.edu.sg}

\author{Zihao Yu}
\affiliation{%
  \institution{Nanyang Technological University}
  \country{Singapore}}
\email{zihao.yu@ntu.edu.sg}

\author{Gao Cong}
\affiliation{%
  \institution{Nanyang Technological University}
  \country{Singapore}}
\email{GAOCONG@ntu.edu.sg}

\renewcommand{\shortauthors}{Weiping Yu, Siqiang Luo, Zihao Yu \& Gao Cong.}

\begin{abstract}
We use machine learning to optimize LSM-tree structure, aiming to reduce the cost of processing various read/write operations. We introduce a new approach {\ML}, which boasts the following features: (1) {\bf ML-Aided}: {\ML} is the first attempt to apply active learning to tune LSM-tree based key-value stores. The learning process is coupled with traditional cost models to improve the training process; (2) {\bf Decoupled Active Learning}: backed by rigorous analysis, {\ML} adopts active learning paradigm based on a decoupled tuning of each parameter, which further accelerates the learning process; (3) {\bf Easy Extrapolation}: {\ML} adopts an effective mechanism to incrementally update the model with the growth of the data size; \textcolor{black}{
(4) {\bf Dynamic Mode}: {\ML} is able to tune LSM-tree online under dynamically changing workloads; (5) {\bf Significant System Improvement}: By integrating {\ML} into a full system RocksDB, the system performance improves by 28\% on average and up to 8x compared to a state-of-the-art RocksDB design.} 
\end{abstract}



\received{January 2024}
\received[revised]{April 2024}
\received[accepted]{May 2024}

\maketitle

\section{Introduction}\label{sec:intro}
\textbf{LSM-Tree based Key-Values Stores.}
Key-value stores, increasingly prevalent in industry, underpin applications in social media~\cite{armstrong2013linkbench,bortnikov2018accordion}, stream and log processing~\cite{cao2013logkv,chen2016realtime}, and file systems~\cite{jannen2015betrfs,shetty2013building}. Notably, platforms like RocksDB~\cite{rocksdb} at Facebook, LevelDB~\cite{google-leveldb} and BigTable~\cite{chang2008bigtable} at Google, HBase~\cite{hbase} and Cassandra~\cite{cassandra} at Apache, X-Engine~\cite{huang2019x} at Alibaba, WiredTiger~\cite{WiredTiger} at MongoDB, and Dynamo~\cite{decandia2007dynamo} at Amazon extensively utilize Log-Structured Merge (LSM) trees~\cite{o1996log} for high-performance data ingestion and fast reads.

An LSM-tree is a multi-level data structure that operates on key-value pairs. The top level of the LSM-tree has a smaller capacity and stores the freshest data, while the lower levels can hold exponentially more data but with progressively older timestamps. Initially, data is inserted into Level-0 (a.k.a. buffer level), until it is full and sort-merged into the next deeper level. The capacity of Level $i+1$ is $T$ times than that of Level $i$, where $T$ is referred to as the size ratio. Similar merge behavior happens in any two consecutive levels, leading to multiple sorted runs in each level, where each run has an associated Bloom filter to facilitate lookups.


\vspace{1mm}
\noindent\textbf{Instance-Optimized LSM-Trees.}
LSM-trees are commonly used in supporting diverse workloads with varying percentages of operation types, such as {\it point lookups}, {\it range lookups} and {\it data writes}. The point lookup (resp. range lookup) is a query that extracts the value (resp. values) corresponding to a given key (resp. key range), whereas the data writes are operations to insert, delete or update the value for a key. The various possibility of the workload raises a crucial question of how to select suitable parameters (e.g., size ratio, compaction policy, memory allocation  between buffer level and Bloom filters) to construct an LSM-tree optimized for a given workload, leading to the notion of instance-optimized LSM-trees. 

Several studies~\cite{huynh2022endure,dayan2018dostoevsky,dayan2019log,huynh2023flexibility} have explored instance-optimized LSM-tree designs, with Dostoevsky~\cite{dayan2018dostoevsky} and K-LSM~\cite{huynh2023flexibility} focusing on compaction policy tuning, LSM-Bush~\cite{dayan2019log} discussing the choice of size ratio between adjacent levels, and Endure~\cite{huynh2022endure} investigating LSM-tree tuning under workload uncertainty. 
\begin{figure}[t!]
\centering
  \includegraphics[width=0.6\linewidth]{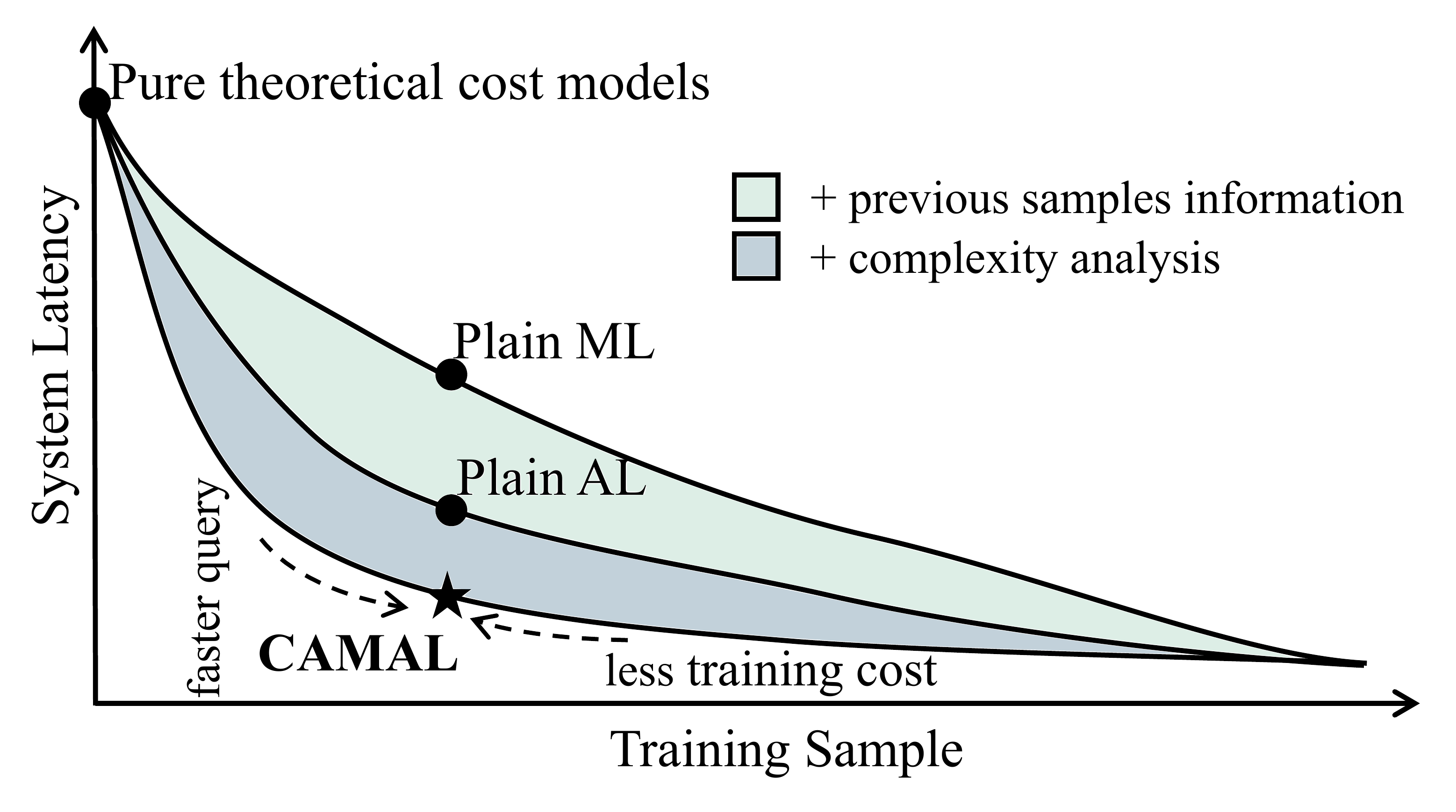}  
  \caption{{Illustrating plain ML approach (e.g., polynomial regression), plain active learning approach and our {\ML} regarding the tradeoff between the training samples and system performance.}
  }
\label{fig:intro}
\end{figure}
All these methods can be classified as {\it complexity-based} models, which primarily rely on complexity analysis to predict the I/O cost of each operation. 
Orthogonal to these approaches, in this paper we ask whether machine learning (ML) can give an even more finer-grained tuning of LSM-tree based key-value stores. The potential of machine learning based methods is that the mapping between system knobs and performance is directly captured in an end-to-end manner, instead of being implicitly obtained via a complexity-based cost model. Such machine-learning aided approach has been proven effective in other data structures such as search trees~\cite{lin2022learning} and spatial data structures~\cite{gu2023rlr}, yet it has been rarely explored for disk-resident data structures such as LSM-trees. It is important to note that the goal of using ML for LSM-tree tuning is not to replace the techniques mentioned above; instead, we aim to explore the power of ML in tuning LSM-trees, 
seeing how to harness the strengths of these two kinds of methods, and this hybrid technique would be one key element in our design. 

\vspace{1mm}
\noindent\textbf{Our Idea: Active learning for instance-optimized LSM-trees.} 
To achieve ML-based tuning, a plain ML approach is expected to train a model, which predicts the LSM-tree performance for a given workload, and is later used for searching a \revise{desired} configuration of the LSM-tree. The model training process starts with collecting samples, followed by feeding the samples to the model for fitting the parameters. A sample is in the form of $(W, X, Y)$, where $W$ is a workload, $X$ is a point in the configuration space formed by all possible ranges for each tunable parameter, and $Y$ is the true running time for workload $W$ using the LSM-tree constructed by parameters $X$. 

{Unfortunately, the cost of sample collection can be prohibitive. Evaluating the system performance $Y$ in a sample requires the ingestion of the entire dataset into the LSM-tree system, followed by the execution of a sufficiently large workload $W$ comprising numerous queries. Consequently, obtaining a single sample with 10GB data may consume several minutes~\cite{huynh2022endure}, with the time frame escalating at an exponentially rapid pace as the dataset grows.} Moreover, given the extensive configuration space, a considerable number of  samples are required to ensure the model's accuracy,  leading to an impractical cost. 

To address the challenge, we turn to the {\it active learning} (AL) paradigm~\cite{ma2020active,mozafari2012active,huang2018optimization}, which has been proven capable of refining the sample quality in an interactive manner, and thus fewer samples are required to secure the model accuracy. For a training workload, AL involves iterative sampling rounds until the given sampling budget is exhausted. Initially the sample set only contains one set of random samples. Each round commences with training an ML model using the existing samples in the set. The model then identifies and selects a new sample that is predicted to have the lowest cost by the newly trained model. This new sample is subsequently added to the sample set, and if the budget permits, the next sampling round is triggered. As the training process advances through each round, the ML model's accuracy progressively improves. The quality of the selected samples also benefits from this refinement over time.

\vspace{1mm}
\noindent
{\bf Challenges and New Designs.} Although active learning effectively refines the exploration in configuration space, we foresee three challenges when applying it to optimize LSM-trees.  Firstly, by default, the initial sample is chosen randomly, which may deviate significantly from the true optimum, requiring additional rounds of exploration. Secondly, the tunable parameters are not effectively decoupled throughout the sampling process, failing to narrow down the configuration space during the process. { Lastly, in the presence of dynamic workloads, re-training becomes necessary and introduces extra costs to maintain model accuracy.}


 The benefits of our approach {\ML} stem from the following novel designs, which address the challenges respectively.   

\vspace{1mm}

 \noindent\textbf{Design 1: Complexity-analysis driven techniques to avoid random initialization.} 
Our insight is that the optimal LSM-tree
parameter settings obtained through a complexity-based
cost model (e.g., the one in~\cite{dayan2017monkey}) still provides much better results than a random sample, although they may not be the true optimum within the configuration space. As a result, this cost model can efficiently and effectively guide active learning to pinpoint the neighborhood containing the true optimum.
Specifically, we can initiate the active learning process with the theoretically optimal solution obtained from the complexity analysis in each training workload, to significantly prune the sampling space for each parameter. 


\vspace{1mm}
\noindent\textbf{Design 2: Decoupling parameters for faster approaching a \revise{desired} solution.}
To reduce the vast sampling space due to the combination of different parameters, we propose a novel {\it hierarchical sampling} technique that decouples each parameter from the complex I/O model. 
We have discovered that although the parameters may be correlated, their optimal settings can be relatively independent, allowing tuning the parameter one round at a time. 
In particular, we establish a theoretical foundation, allowing us to first assess the \revise{desired} values of one parameter, and then address more intricate dependencies of other parameters to guide the sampling rounds in active learning.


 \vspace{1mm}
\noindent\textbf{Design 3: Extrapolation strategy for data growth.} While decoupled active learning effectively narrows the sampling space, the training cost tends to rise with larger data sizes. To accommodate data growth, we have theoretically proven that it is possible to rapidly transition to new \revise{tuned} parameters from existing ones without retraining. This method helps avoid the exponential increase in training costs, leading to greater efficiency.

 \vspace{1mm}{
\noindent\textbf{Design 4: Dynamic mode.}
In practical scenarios, system faces dynamically changing workloads. 
We then design the dynamic LSM-tree based on the extrapolation strategy to combat the challenge, see Section~\ref{sec:dynamic}.
}


 \vspace{1mm}

\noindent\textbf{Contributions.} In summary, Figure~\ref{fig:intro} illustrates {\ML}'s position relative to other ML approaches based on our experimental evaluation in Section~\ref{sec:eva}. Our contributions are summarized as follows.
\begin{itemize}[itemsep=5pt,leftmargin=*]
\item We propose a new model named {\ML}, the first attempt to apply active learning for LSM-tree instance optimizations. It integrates the complexity-based models to align active learning to the LSM-tree context 
(\S~\ref{sec:overview}).

\item We present a novel hierarchical sampling technique, which is particularly designed for LSM-tree applications, to reduce the sampling space, {significantly cutting down training time and improving practical usability} 
(\S~\ref{sec:theory}).

\item Our model {\ML} embraces data growth, in that it  reasonably extrapolates the \revise{desired} settings without retraining 
(\S~\ref{sec:extrap}).

\item 
{
Equipped with the extrapolation strategy, we have enhanced {\ML} to handle dynamically changing scenarios. We further introduce a novel design named DLSM, an LSM-tree variant specifically engineered to adapt to dynamic workloads (\S~\ref{sec:dynamic}).}

\item  
We examine three widely used ML models to be embedded into {\ML} and discuss their benefits and drawbacks (\S~\ref{sec:ml}).

\item We integrated our method with the widely-used LSM key-value database, RocksDB, to demonstrate its practicality. Our approach can significantly reduce latency by up to 8x, compared to a state-of-the-art RocksDB design 
(\S~\ref{sec:eva}).

\end{itemize}

\section{Background}
\label{sec:background}

\begin{figure*}
\raisebox{29mm}{
\hspace{-18mm}
    \begin{subfigure}{0.4\textwidth}
        \centering
        \label{table:terms}
\renewcommand{\arraystretch}{1.02}
    \setlength{\tabcolsep}{3pt} 
\footnotesize
\begin{tabular}{l||l||l}
\textbf{Term} & \textbf{Definition}     & \textbf{Unit} \\ \hline\hline
$N$    & total number of entries & entries       \\ 
$L$    & number of levels        &   levels            \\
$B$    & number of entries that fit in a storage block&  entries \\ 
$E$    &size of a key-value entry   &       bits       \\ 
$T$    &size ratio between adjacent levels   &               \\ 
$T_{lim}$  &size ratio at which $L$ converges to 1  &         \\ 
$M_b$  &memory allocated for write buffer  &   bits      \\ 
$M_f$  &memory allocated for bloom filters  &   bits      \\ 
$M_c$  &memory allocated for block cache  &   bits      \\ 
$M$  &total memory ($M_b+M_f+M_c$) &    bits    \\ 
$v$  &percentage of zero-result point lookups &             \\ 
$r$  &percentage of non-zero-result point lookups &             \\ 
$q$  &percentage of range lookups &             \\ 
$w$  &percentage of writes &             \\ 
$s$  &selectivity of range lookups &  entries           \\ 
\end{tabular}
    \end{subfigure}}
    \hspace{19mm}
     \begin{subfigure}{0.35\textwidth}
        \centering
          \label{tab:cost}
  \renewcommand{\arraystretch}{1.4}
    \footnotesize
  \begin{tabular}{l||l||l} 
    \multicolumn{1}{l||}{\textbf{Operation}} & \multicolumn{1}{l||}{\textbf{Leveling}} & \textbf{Tiering} \\ \hline\hline
\multicolumn{1}{m{2cm}||}{zero-result point lookups ($V$)}
&   $e^{-\frac{M_f}{N}}$  & $e^{-\frac{M_f}{N}}\cdot T$       \\
\multicolumn{1}{m{2cm}||}{non-zero-result point lookups ($R$)}
&  $e^{-\frac{M_f}{N}}+1$  &  $e^{-\frac{M_f}{N}}\cdot T+1$ \\
range lookups  ($Q$)   & $L+\frac{s}{B}$ & $L\cdot T+\frac{T\cdot s}{B}$\\
writes ($W$)  &    $\frac{L\cdot T}{B}$ &    $\frac{L}{B}$ \\
\end{tabular}

\vspace{4mm}

  \includegraphics[width=1.35\textwidth]{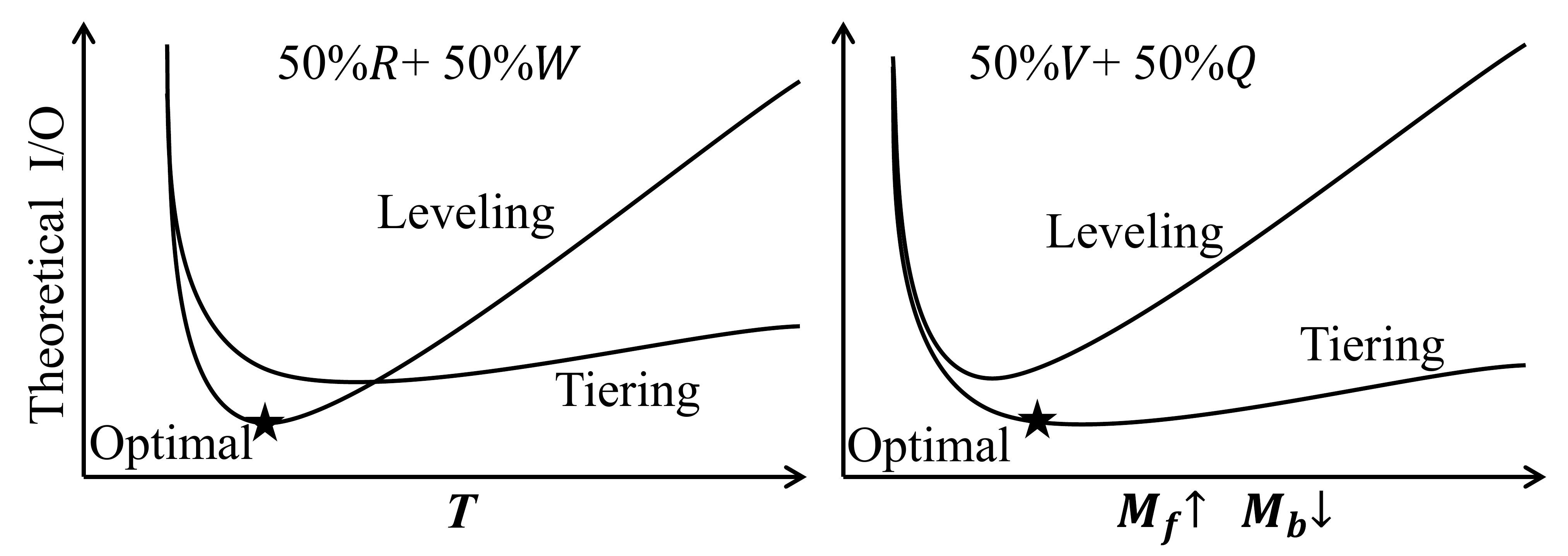}  
\end{subfigure}
      \caption{Parameters of LSM-trees and workloads used throughout the paper, complexity-based I/O cost models, and base sampling space of instance-optimized LSM-trees.}
    \label{fig:example}
\end{figure*}

\vspace{1mm}
\noindent\textbf{LSM-tree Structure.} An LSM-tree is structured with multiple levels, where each level contains one or multiple sorted runs. New data is initially stored in an unsorted state in a memory buffer, whose filled-up will trigger a write to the first level of the LSM-tree. As each level becomes full, its data is sorted and merged into the next level recursively. During point lookup, the LSM-tree searches the most recent level first, followed by lower levels until it finds the matching data. Additionally, Bloom filters optimize point lookup by efficiently determining if a key exists in a sorted run without performing I/Os. The capacity of each level in an LSM-tree grows by a size ratio of $T$. Therefore, the number of levels $L$ is determined by the size ratio $T$, the memory allocated to the write buffer $M_b$, and the size of the data. Level $i$ has $\frac{M_b}{E} \cdot (T-1)T^{i-1} $ entries~\cite{huynh2022endure}, where $E$ represents the size of an entry. If there are $N$ entries in the LSM-tree, the number of levels can be calculated as:
\begin{equation}
\label{equ:level}
    L=\left\lceil \log_T (\frac{N\cdot E}{M_b}+1)  \right\rceil
\end{equation}
In line with Dostoevsky~\cite{dayan2018dostoevsky}, we restrict the size ratio range to $2 \leq T \leq T_{lim}$, where $T_{lim}$ is defined as ${N\cdot E}/{M_b^{lim}}$. 

\vspace{1mm}
\noindent\textbf{Workload.} Following previous works~\cite{dayan2017monkey,dayan2018dostoevsky,huynh2022endure}, key-value databases commonly involve four types of operations: zero-result point lookups, non-zero-result point lookups, range lookups, and writes. In the case of point lookups, a single key is searched in the database, and the Bloom filter is often used to determine whether the key is located in the block before proceeding with the actual block reading. Point reads can be further divided based on whether the key exists in the database, categorized as zero-result point lookups and non-zero-result point lookups. Range lookups aim to find the most recent versions of all entries within a specified key range. This process entails merging the relevant key range across all runs at every level while sorting the entries. 
Writes, including inserts, deletes, and modifications, typically involve appending a new key in the write buffer instead of locating the older version for in-place updates. 
In practice, workloads for key-value databases are often comprised of varying proportions of the four operations. 

\vspace{1mm}
\noindent\textbf{Complexity-analysis based Model.} The complexity-analysis based Model, or theoretical I/O model, analyzes the number of I/Os generated per query, based on a given workload and basic parameters of an LSM-tree. We adopt the state-of-the-art Bloom filter bits-allocation strategy proposed in Monkey~\cite{dayan2017monkey}, which designs different false positive rates (FPRs) for Bloom filters in different levels to achieve optimal memory allocation. We also follow the models derived in Monkey~\cite{dayan2017monkey} for the four workload types, as shown in Figure~\ref{fig:example}. These models have been widely recognized and adopted in later works, such as Dostoevsky~\cite{dayan2018dostoevsky} and Endure~\cite{huynh2022endure}. 
{It is important to note that although we use a relatively simple complexity-based model, it serves the purpose in our hybrid framework that combines the complexity-based model and ML model because the goal of using the complexity-based model is to {\it estimate} a reasonable range. Our framework can integrate with a more sophisticated model when necessary. 
}
According to the cost model of each operation in Figure~\ref{fig:example}, the overall average cost can be calculated based on the known proportions ($v,r,q,w$) of the four operations:
\begin{equation}
    f = v \cdot V + r \cdot R + q \cdot Q+ w \cdot W
    \label{equ:cost}
\end{equation}

\noindent\textbf{Objectives.}
This paper aims to minimize the end-to-end latency incurred by the input workload, given the knowledge of the system and workload. We consider both the static mode (Sections~\ref{sec:overview}-\ref{sec:extrap}) where the workload is stable, and the dynamic mode (Section~\ref{sec:dynamic}) where the workload changes online. Figure~\ref{fig:example} provides examples of how the performance varies across different workloads based on specific parameters. The system includes information such as the number of entries $N$, the total memory budget $M$, the size of an entry $E$, and the number of entries $B$ that fit in a storage block. The workload mainly includes the proportion of each operation ($v,r,q,w$). The tunable parameters we consider include the size ratio $T$, the memory budget $M_b$ for the writer buffer (Level 0 of the LSM-tree), the memory budget $M_{f}$ for the Bloom filter, and the compaction policy (leveling or tiering). 

\revise{The reason we select end-to-end latency as the optimization objective is that it is one of the most straightforward metrics to reflect the system performance, evaluated in many studies~\cite{dayan2018dostoevsky,mo2023learning,dayan2017monkey}. It is equivalent to maximizing throughput since it represents the number of operations a system completes in a period, which can be derived from the reciprocal of end-to-end latency. The I/O cost, which counts the I/Os triggered by an operation, is sometimes not general enough to indicate overall system performance, as it excludes the influence of devices and background jobs. We will discuss this further in the evaluation part.}

\revise{In selecting the tunable parameters, we mainly consider their high impact in performance~\cite{huynh2022endure,dayan2017monkey,dayan2018dostoevsky} and 
their prevalence in practical applications in LSM-based storage systems~\cite{rocksdb,google-leveldb}. Our primary aim is to explore the potential of using active learning to tune LSM-trees, and thus we select these relatively representative tunable parameters for a pioneering study. We acknowledge that incorporating the tuning parameters within a broader scope (e.g., sub-compaction, SST file size, background threads, cache policies) could further optimize the LSM-trees. However, adding more parameters also increases the complexity of training and prediction, which may not necessarily enhance overall performance. For example, while adjusting level-based compaction policies~\cite{mo2023learning} to assign varying numbers of runs ($K_i$) to each level could improve system performance, the expansion of the parameter space could lead to prohibitive sampling costs. Moreover, many parameters are not universally applicable across all LSM systems (e.g., sub-compaction is specific to RocksDB). Therefore, we limit our parameter tuning to showcase the idea of our proposed model. Nonetheless, we still provide an evaluation on a broader range of parameters to demonstrate how to extend our model in the evaluation section.}
\section{{{\ML}} Overview}
\label{sec:overview}


\begin{table}[tb!]
\footnotesize
\centering
\begin{center}
  \caption{Operation percentages in 15 training workloads.}
  \label{tab:train_workload}
  \begin{tabular}{l||lllllllllllllll} 
\textbf{No.} & \textbf{1} & \textbf{2} & \textbf{3} & \textbf{4} & \textbf{5} & \textbf{6} & \textbf{7} & \textbf{8} & \textbf{9} & \textbf{10} & \textbf{11} & \textbf{12} & \textbf{13} & \textbf{14} & \textbf{15} \\
 \hline \hline
   \textbf{$v$}: & 25 & 97 & 1 & 1 & 1 & 49 & 49 & 49 & 1 & 1 & 1 & 33 & 33 & 33 & 1\\
   $r$: & 25 & 1 & 97 & 1 & 1 & 49 & 1 & 1 & 49 & 49 & 1 & 33 & 33 & 1 & 33\\
  $q$: & 25 & 1 & 1 & 97 & 1 & 1 & 49 & 1 & 49 & 1 & 49 & 33 & 1 & 33 & 33\\
   $w$: & 25 & 1 & 1 & 1 & 97 & 1 & 1 & 49 & 1 & 49 & 49 & 1 & 33 & 33 & 33\\ \cmidrule(lr){3-6}\cmidrule(lr){7-12}\cmidrule(lr){13-16}
\multicolumn{2}{r}{} & \multicolumn{4}{c}{\raisebox{1.5ex}{\textbf{unimodal}}} 
 & \multicolumn{6}{c}{\raisebox{1.5ex}{\textbf{bimodal}}}
  & \multicolumn{4}{c}{\raisebox{1.5ex}{\textbf{trimodal}}}\\
\end{tabular}
\end{center}
\end{table}

{
\small 
\begin{algorithm} \DontPrintSemicolon
\KwIn{Workloads to be tuned $\mathcal{W}$} 
\KwOut{Tuned configurations $\mathcal{C}$}
Collect training samples $\mathcal{S}$ by decoupled active learning (see Section~\ref{sec:theory} and Algorithm~\ref{alg:decopuled}). \;
Train ML cost models with $\mathcal{S}$ (see Section~\ref{sec:ml}).\;
Get tuned configurations $\mathcal{C}$ from ML models.\;
\If{$\mathcal{W}$ is scaled-up}{
Get $\mathcal{C}$ with extrapolation strategy (see Section~\ref{sec:extrap}).\;
}
\If{$\mathcal{W}$ is dynamic}{
Get $\mathcal{C}$ with dynamic mode (see Section~\ref{sec:dynamic}).\;
}
\Return $\mathcal{C}$
\caption{\revise{{\ML} overview}}
\label{alg:overview}
\end{algorithm}
}

\revise{
This section introduces {\ML}, outlined in Algorithm~\ref{alg:overview}, where an ML model estimates LSM-tree costs for given workloads and identifies desired parameter settings for enhanced performance.} As shown in Figure~\ref{sec:overview}, in its training phase, {\ML} first decouples the sampling space and identifies the theoretical optimum using a complexity-based model. Following this, it integrates an ML model to facilitate an active learning cycle, where the model is continuously iterated to select subsequent samples. To address data growth, we introduce an extrapolation strategy that extends the \revise{desired parameters} to testing scenarios with larger data sizes without the need for retraining. Both of these methods are designed to reduce training costs. Additionally, to meet practical demands, we apply the extrapolation strategy to enhance {\ML} for dynamic environments, which also includes equipping LSM-trees with the ability to dynamically adapt to changing parameters.

{\ML} considers training with various workloads as shown in Table~\ref{tab:train_workload}, in line with settings in Endure~\cite{huynh2022endure}. For ease of discussion, let us focus on one workload $W$ with a sampling budget $h$, as extending to multiple workloads is straightforward -- simply training one workload after another. As shown in Figure~\ref{fig:frame}, the main workflow of {\ML} follows an active-learning approach, which consists of multiple rounds of sample generation and model training using existing samples.  
At a high level, {\ML} employs a novel technique called {\it decoupled active learning}, which enables the tuning of individual parameters in separate rounds. The parameter is fixed to its \revise{tuned} value found based on an intermediate 
model trained during active learning. This technique addresses a limitation commonly found in typical active learning, where all parameters are sampled together in a correlated manner, resulting in an extensive sampling space. Assuming the sampling space size for each parameter is $n_i$, the total sampling space required to explore all parameter correlations is $\prod n_i$. In contrast, decoupled sampling allows us to tune each parameter separately in a round, reducing the sample space to $\sum n_i$. Additionally, by using decoupled sampling, the model gains early exposure to \revise{desired} settings, leading to higher-quality samples in subsequent rounds.
\begin{figure*}[t!]
\centering
  \includegraphics[width=\linewidth]{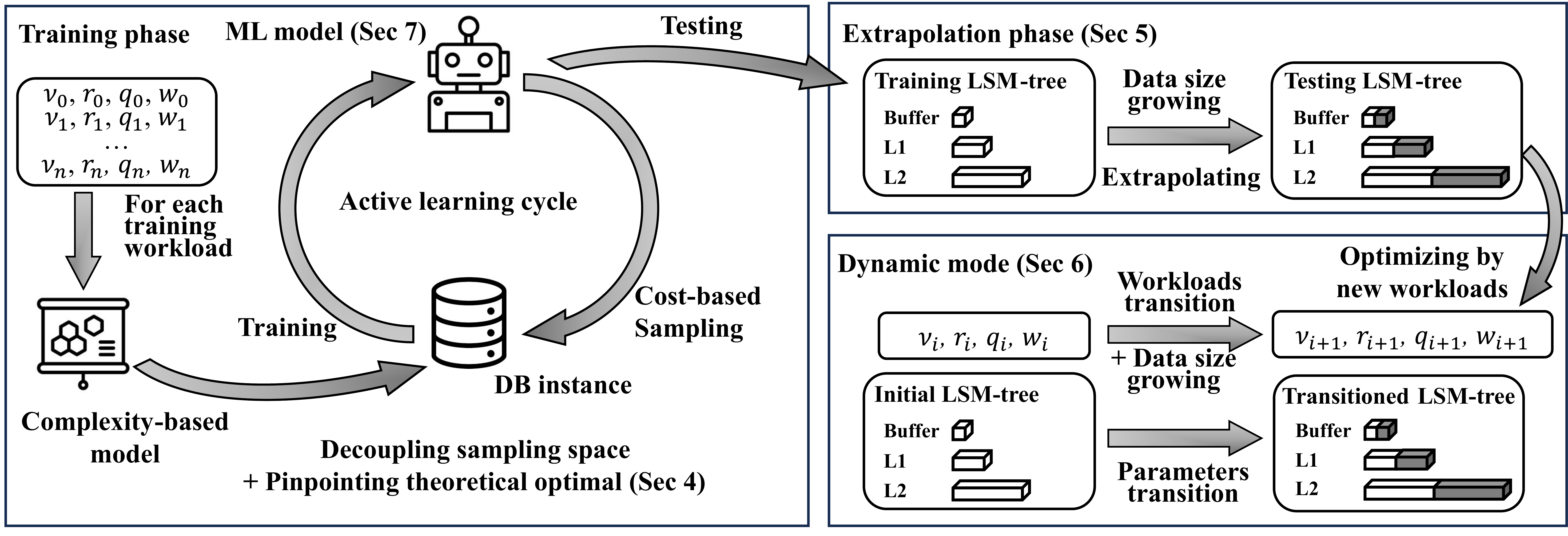}  
  \caption{
The overview of {\ML}: in its training phase, {\ML} first decouples the sampling space and identifies the theoretical optimum using a complexity-based model. Following this, it integrates an ML model to facilitate an active learning cycle, where the model is continuously iterated to select subsequent samples. To address data growth, we introduce an extrapolation strategy that extends the training optimals for larger data sizes without retraining. Both of these methods are designed to reduce training costs. Additionally, to meet practical demands, we apply the extrapolation strategy to enhance {\ML} for dynamic environments, which also includes equipping LSM-trees with the ability to dynamically adapt to changing parameters.}
\label{fig:frame}
\end{figure*}
One may wonder why the \revise{desired} setting for a parameter can be determined by an intermediate model in AL. Addressing this query leads us to another crucial technique in {\ML}, which involves employing a complexity-based cost model to assess the independence between \revise{desired} settings for each parameter, and is particularly geared to the LSM-tree applications. 
In a nutshell, each round in {\ML} comprises the following steps. Initially, a parameter or a set of parameters to be adjusted in the current round is chosen. Next, an analysis is performed on the \revise{selected} parameter configuration using a complexity-based cost model, and we narrow down the sampling range to the neighborhood of the \revise{selected} value.
Subsequently, the parameter is sampled within this narrowed range while keeping the previously selected parameters intact, and other parameters are set to their default values. Then, the selected parameters, denoted by $X$, are used to construct the LSM-tree instance and record the performance (denoted by $Y$) of running the workload $W$ to form the sample $(W, X, Y)$. Afterwards, a machine learning model (will be discussed in Section~\ref{sec:ml}) is trained using all existing samples. Finally, the trained model is used to tune the current parameter by selecting the one that leads to the lowest inference value for the model. 
The selected parameter value is then fixed for subsequent rounds to mitigate potential errors arising from complexity-based analysis. 

{
While decoupled active learning effectively reduces the sampling space for moderate data sizes, the training cost can exponentially increase with larger data sizes. 
To address this challenge, we design an extrapolation strategy that scales the \revise{selected} settings from a smaller to a larger dataset without retraining. Specifically, if $T'$, $M'_f$, and $M'_b$ are the \revise{selected} parameters for a database with $N'$ entries and an $M'$ memory budget, we prove that when the data size grows to $kN'$ and the memory budget to $kM'$, the \revise{desired} parameters become $T'$, $kM'_f$, and $kM'_b$. Though the factor $k$ may have practical limitations, this strategy still substantially reduces training costs, cutting down training time by approximately an order of magnitude. 
}

{
In practical applications, another challenge is managing dynamic workloads where the \revise{desired} settings for an initial workload may not remain effective for subsequent ones. To illustrate this scenario, we have a current workload represented by $(v_i,r_i,q_i,w_i)$ and an expected next workload $(v_{i+1},r_{i+1},q_{i+1},w_{i+1})$. In the interval between these two workloads, there are enough updates to trigger sufficient compactions, allowing for a change in parameters. Our method involves initially setting the \revise{desired} parameters $T'$, $M'_f$, and $M'_b$ for workload $i$. We then use the extrapolation strategy to estimate the \revise{desired} parameters $T^{\prime\prime}$, $M^{\prime\prime}_f$, and $M^{\prime\prime}_b$ for workload $i+1$. In the interval between these workloads, we incrementally adjust the size ratio during each compaction, and change the bit per key for Bloom filters when new runs are formed. This approach enables the settings to gradually shift, aligning with the evolving workloads and ensuring continued optimization.}

\section{Decoupled Active Learning}\label{sec:comp}
\label{sec:theory}
{
\small
\begin{algorithm}
\DontPrintSemicolon
\KwIn{Set of training workloads $\mathcal{W} = \{  {w_i} \} $ \;
  }
Initialize the set of training samples $\mathcal{S} \gets \{ \}$\;
\ForEach{ $w_i  \in \mathcal{W} $ }{
    Get theoretical optimal size ratio $T^*$ \;
    Get neighbour range $\mathcal{T} $ around $T^*$ \;
    \ForEach{ $T_i  \in \mathcal{T} $}{
        Execute database instance with $T_i$ \;
        $\mathcal{S} \gets \mathcal{S} \wedge (w_i,x_i,y_i)$ \;
        /* $x_i$ is LSM-tree parameters*/  \;
        /* $y_i$ is system latency */ \;
    }
    Train the ML model using $\mathcal{S} $  \;
    Get practical \revise{desired} $T'$ by the ML model  \;
    Get theoretical optimal buffer $M^*_b$ with fixed $T'$ \;
    Get neighbour range $\mathcal{M}$ around $M^*_b$ \;
    \ForEach{ $M_i  \in \mathcal{M} $}{
        Execute database instance with $M_i$ \;
        $\mathcal{S} \gets \mathcal{S} \wedge (w_i,x_i,y_i)$ \;
    }
}
\Return $\mathcal{S}$  \;
\caption{Decoupled active learning}
\label{alg:decopuled}
\end{algorithm}
}

Decoupled active learning aims to decouple each individual parameter from the complex I/O model to ease the sampling of a single parameter at each round. In particular, given a configuration space $\Lambda$ that consists of all tunable parameters, we select a parameter $\lambda_i\in \Lambda$ to analyze first, if there exists an optimal solution  $\lambda^*_i$ for $\lambda_i$, such that
\begin{align}
\begin{cases}
    \enspace\, \left. \frac{\partial f}{\partial \lambda_i}\right|_{\lambda_i=\lambda^*_i}=0\\
    \quad \lambda^*_i \perp {\Lambda \backslash \lambda_i}
    \end{cases}
\end{align}
where $f$ is the cost function, and $\perp$ denotes an orthogonal relationship. Essentially, the first equation implies that $\lambda^*_i$ is an optimal solution to $\lambda_i$ for the cost function $f(\cdot)$\footnote{The partial derivative at the optimal value is zero.}, whereas the second expression indicates that the optimal solution $\lambda^*_i$ does not depend on other tunable parameters in the configuration space. As such, we can safely determine the optimal setting for parameter $\lambda^*_i$ first (in {\ML}, we use machine learning to help calibrate the setting, see Section~\ref{sec:ml}) and then shrink the configuration space to $\Lambda^*=\Lambda \backslash \lambda^*$ by eliminating the dimension of $\lambda_i$ (which has been fixed as $\lambda^*_i$). The process then repeats in the newly reduced dimension $\Lambda^*$.

The process for setting each parameter involves several key steps. Initially, we formulate a cost function related to the parameter. Next, we proceed to solve the derivative of this function, obtaining a theoretical optimum. This theoretical optimum establishes an effective sampling range centered around the optimum to encompass the practical desired solution close to this range. The underlying theoretical basis for this approach is based on the following lemma:
\begin{lemma}
Given the prevalent cost model~\cite{dayan2018dostoevsky,dayan2017monkey} for leveling, the process of configuration optimization can be decoupled into two distinct stages: firstly, determining the optimal value of $T^*$, and secondly, allocating memory between $M_b$ and $M_f$. This decoupling ensures the attainment of a globally optimal configuration combination.
\end{lemma}
\begin{proof}
According the terms in Figure~\ref{fig:example}, the I/O cost for a leveling LSM-tree can be represented as:
\begin{equation}
  \begin{aligned}
    f_l(T)&= ve^{-\frac{M_f}{N}}+r(e^{-\frac{M_f}{N}}+1)+ q(L+\frac{s}{B})+ w\frac{L \cdot T}{B} \nonumber\\
  \end{aligned}
\end{equation}
Its derivative can be formulated as:
\begin{equation}
    \label{equ:l_t_diff}
  \begin{aligned}
    \frac{\partial f_l}{\partial T}&= \frac{L}{BT\log(T)} \cdot (wT(\log{T}-1)-qB)\\
  \end{aligned}
\end{equation}
The theoretical optimal $T^*$ (the optimal value of $T$) can be obtained by setting Equation~\ref{equ:l_t_diff} to zero \revise{according to the Derivative Test}, giving us 
\begin{equation}
    \begin{aligned}
    wT^*(\log{T^*}-1)-qB=0\label{equ:l_t_diff_2}
    \end{aligned}
\end{equation}
Clearly, from Equation~\ref{equ:l_t_diff_2} we conclude that $T^*$ is independent of $M_f$ and $M_b$. So, we can determine $T^*$ first. 

When optimizing for $M_f$ and $M_b$, we have:
{
\begin{equation}
  \begin{aligned}
    \frac{\partial f_l(M_f)}{\partial M_f}&= -\frac{1}{N}(v+r)e^{-\frac{M_f}{N}}+\frac{qB+wT}{B\log{T}}\frac{1}{(M-M_f)}\\
  \end{aligned}
    \label{equ:l_filt_diff}
\end{equation}
}Equation~\ref{equ:l_filt_diff} implies that when $M_f$ increases, the absolute value of the first term will decrease, thereby diminishing its effect on optimizing point lookups (whose cost is $(v+r)e^{-{M_f}/{N}}+r$). Concurrently, a larger $M_f$ will result in a reduced $M_b$, indicating that the second term will increase, further causing the cost to escalate at an accelerating rate. To strike a balance between the two parameters, we set Equation~\ref{equ:l_filt_diff} equal to zero, which will allow us to derive the theoretically optimal values for $M_f$ and $M_b$. This completes the proof of the lemma.
\end{proof}

Upon segregating the parameters across various rounds, during the sampling process, we gather samples from a contiguous vicinity centered around the theoretical optimal values, based on the specified sample count. For $T$, the sampling interval is an integer. For $M_f$, the interval is measured in bits per key (BPK), determined by the formula $\frac{M_f^*}{N}$. 

{
In summary, decoupled active learning is outlined as Algorithm~\ref{alg:decopuled}. It stands out compared to plain active learning primarily for two reasons: firstly, it identifies theoretical optimal parameters, which are often close to the practical \revise{desired} parameters. Secondly, it reduces the dimensions of the sampling space, significantly decreasing the number of active learning cycles required to find the practical \revise{desired} solutions.}

{
\vspace{1mm}
\noindent
{\bf Extension to Tiered LSM-tree.} The decoupled sampling approach remains largely pertinent for the tiering compaction policy. First, we can differentiate with respect to $T$ using the tiering cost model. This differentiation closely resembles Equation~\ref{equ:l_t_diff}, augmented with a specific term representing the point read cost:
\begin{equation}
\label{equ:tier_t_diff}
\begin{aligned}
\frac{\partial f_t}{\partial T}&= (v+r)e^{-\frac{M_f}{N}}+q\frac{s}{B}+\frac{L(qBT(\log T-1)-w)}{BT\log{T}} \
\end{aligned}
\end{equation}
Based on prior research~\cite{huynh2022endure}, this term is quantitatively less significant compared to I/Os of range lookups and post-optimization writes. Therefore, in practice, the \revise{desired} $T^*$  exhibits only mild fluctuations with changes in $M_f$. Furthermore, employing machine learning for subsequent sampling effectively addresses these variations, as the sampling phase refines and corrects inaccuracies.
}

\section{Extrapolation without Retraining}
\label{sec:extrap}{
Employing ML models for estimating the cost of an LSM-tree introduces a related issue concerning their extrapolation capabilities. This means that when testing configurations deviate from those used during training, it becomes important to determine if the stale model can still effectively optimize the LSM-tree. For instance, if the stale model is trained on a database with $N=10^6$, it is unclear whether it can be applied to a database with $N=10^7$. While retraining the model is a possibility, we aim to identify an incrementally updated solution based on the existing model, which does not require a costly retraining. 
}

\begin{lemma}
Given the optimal $T'$ and $M'_f$ under the memory budget $M'$ and the number of entries $N'$, we have the new optimal $T''$ and $M''_f$ under the memory budget $M''=kM'$ and the number of entries $N''=kN'$ as:
\begin{equation}
  \begin{aligned}
    T''&=T',\quad\quad
    M''_f    &= kM'_f
  \end{aligned}
    \label{equ:extrap_lemma}
\end{equation}
\label{lemma:2}
\end{lemma}
\begin{proof}
{
According to~\cite{idreos2019design}, in leveling policy, the write amplification of each level is $T$, signaling that on average the key-value entry of the update will take part in $T$ compactions in each level. So the average complexity of CPU to compact an entry is expected to be $T\cdot L$. Further separating the I/O cost model in~Figure~\ref{fig:example} by writes and reads, on average the write query incurs a total overhead of $\frac{LT}{B} \cdot (I_w+I_r)+C_w\cdot T$, where $I_w$ and $I_r$ are the write and read I/O costs while $C_w$ is the CPU cost, such as merge sorting and space allocation. In the same way, let $C_r$ be the cost of probing the metadata of a sorted run in the main memory, then the total cost of an zero-result point lookup is expected to be $e^{-\frac{M_f}{N}} \cdot I_r+C_r\cdot L$, and the cost for a non-zero-result point lookup is $(e^{-\frac{M_f}{N}}+1) \cdot I_r+C_r\cdot L$. For range lookups, the main CPU overhead is also to retrieve the information in the metadata, which leads to a total overhead $(L+\frac{s}{B})\cdot I_r+C_q\cdot L$.
}

{
In summary, when the number of entries is $N'$ and the memory budget is $M'$, the total overhead per operation of leveling policy is
\begin{equation}
  \begin{aligned}
    g_l&=I_rve^{-\frac{M_f}{N'}}+I_rr(e^{-\frac{M_f}{N'}}+1) +2C_r  L\\
      &+ I_rq(L+\frac{s}{B}) +C_q L + (I_w+I_r)\frac{LT}{B}  +C_w TL
  \end{aligned}
    \label{equ:extrap_orig}
\end{equation}
If the overhead is minimized when $T=T'$ and $M_f=M'_f$, then we have the equations:
\begin{equation}   
\label{equ:extrap_opt}
  \begin{aligned}
   \log{\frac{N'E}{M'-M'_f}} \frac{2C_r+I_rq+C_q+T'(T'-\log{T'})(I_w+I_r+C_w)}{T'\log^2{T'}} =0 \\
   -\frac{I_r}{N'}(v+r)e^{-\frac{M'_f}{N'}}+\frac{2C_r+I_rq+C_q+(I_w+I_r)\frac{T'}{B}+C_wT}{\log{T}(M'-M'_f)}=0
  \end{aligned}
\end{equation}
When $N''=kN'$, to ensure that the equation still holds, we can make $T^{"}=T'$ and $M'_f=kM'_f$, then the new version of Equations~\ref{equ:extrap_opt} is transformed to be
\begin{equation}
  \begin{aligned}
   \log{\frac{kN'E}{kM'-kM'_f}} \frac{2C_r+I_rq+C_q+T'(T'-\log{T'})(I_w+I_r+C_w)}{T'\log^2{T'}} =0 \\
   -\frac{I_r}{kN'}(v+r)e^{-\frac{kM'_f}{kN'}}+\frac{2C_r+I_rq+C_q+(I_w+I_r)\frac{T'}{B}+C_wT}{\log{T}(kM'-kM'_f)}=0,\nonumber
  \end{aligned}
    \label{equ:extrap_opt2}
\end{equation}
which means $T^{\prime\prime}=T'$ and $M^{\prime\prime}_f=kM'_f$ are the new optimum of the new configuration.
} This completes the proof.
\end{proof}

Based on Lemma~\ref{lemma:2}, the extrapolation strategies of $N$ can be inducted as: First, we train the model with configuration of $N'$ and $M'$.
Then we get a stale optimal $T'$ and $M_f^{\prime}$ under given workloads. Then we get  $T''$ and $M_f''$ from Equation~\ref{equ:extrap_orig} as the new optimal when $N''=kN'$ and $M''=kM'$. 

\section{Dynamic System Mode}\label{sec:dynamic}
\begin{figure*}[t!]
\centering
  \includegraphics[width=\linewidth]{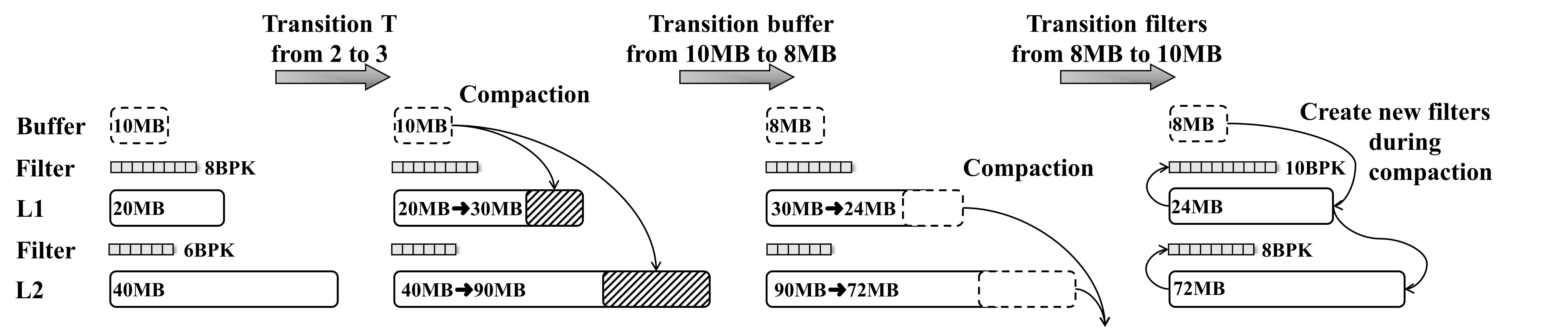}  
  \caption{Running example of the dynamic system mode. \revise{We employ a lazy transition strategy to keep the transition cost lower.}}
\label{fig:change}
\end{figure*}
{
Workloads can be dynamic. This section discusses how to adapt {\ML} to workload shifts. Assume that we have a current workload represented by $(v_i,r_i,q_i,w_i)$ and an expected next workload $(v_{i+1},r_{i+1},q_{i+1},w_{i+1})$. 
Our method initially sets the \revise{desired} parameters $T'$, $M'_f$, and $M'_b$ for workload $i$, and then uses the extrapolation strategy to estimate the \revise{desired}  parameters $T^{\prime\prime}$, $M^{\prime\prime}_f$, and $M^{\prime\prime}_b$ for workload $i+1$. In the interval between these workloads, we incrementally adjust the size ratio at each compaction, and change the bit per key for Bloom filters when creating new runs. This approach enables a gradual shift of the settings, aligning with the evolving workloads and ensuring continued optimization.}

To achieve the dynamicity, we propose a dynamic LSM-tree design, capable of \revise{lazy} dynamically adjusting the size ratio and memory allocation. Specifically, regarding the size ratio and memory buffer, their impact manifests chiefly through their ability to reshape each level's capacity. This changing process is guided by a common principle: 

\begin{itemize}[itemsep=5pt,leftmargin=*]
\item{Firstly, in cases where the target size for a level is less than its present size, we employ a strategy of compacting data from shallower levels into deeper ones to align with the desired level size. This is done during standard compactions. As a result, this often means that when compaction is triggered, some files from the current level are moved to the next level, also leading to more frequent compaction.
For example, Figure~\ref{fig:change} illustrates four states of the dynamic LSM-tree, reflecting the process of parameter changes during LSM-tree compactions in response to workload changes. Initially, the size ratio is 2, and the memory buffer is allocated 10MB. Level 1 (L1) contains 20MB and has a filter of 8 BPK, while Level 2 (L2) holds 40MB, also with a filter of 6 BPK. In the second state, we see an adjustment in the size ratio, with L1 expanded to 30MB and L2 to 90MB, suggesting compactions have occurred allowing for increased data volume per level. 
}

\item{Secondly, if the desired size for a level is greater than its current size, possibly due to an increase \revise{desired} size ratio or write buffer, the level is gradually expanded using data from shallower levels during regular compaction processes. This typically results in less frequent compaction since each level has the capacity to hold more files. In the third scenario depicted in Figure~\ref{fig:change}, the buffer size is reduced to 8MB, with L1 now at 24MB and L2 at 72MB. The remaining files are compacted into deeper levels, adapting the structure to better suit the current workload.  }
\end{itemize}
In practice, the actual size of each level may deviate from the ideal configuration dictated by the predefined size ratio and write buffer due to \revise{the lazy transition strategy}. However, when workloads evolve gradually, the expected optimal parameters transition smoothly as well. This incremental adaptation allows the actual parameters to closely align with the \revise{desired value} over time. It can be validated in our experiments that the dynamic LSM-tree is engineered to maintain better optimization performance than plain LSM-trees.
For altering the memory allocated to the bloom filter, we allocate the new bits-per-key calculated from $M'_f$ to each file as it is created. In the final stage of Figure~\ref{fig:change}, the adjustment lets the filter size be increased to 10MB. This happens when each new Bloom filter is created during compaction. For instance, the filter at L1 is enhanced to 10 BPK. This adjustment signifies a refinement of the filters, optimizing them for more efficient key lookups in the context of the changing data environment.

\revise{For automatic optimization in dynamic scenarios, we employ threshold-based detection methods. Specifically, to recognize workload changes, the system monitors the percentages of each operation within a {\it period} of $p$ operations, where $p$ is a hyper-parameter. Reconfiguration does not occur in every period, but only when the percentage for any operation type varies by a predefined threshold $\tau$ compared to its percentage at the time the last reconfiguration was triggered. The system then handles the reconfiguration according to the approach elaborated above. Additionally, we evaluate the sensitivity of $p$ and $\tau$ in Section~\ref{sec:eva}.}

\section{Embedded Machine Learning Models}
\label{sec:ml}

This section examines the three most prevalent ML models, evaluating their pros and cons for integration into {\ML}. It is nevertheless to note that other ML models can potentially be embedded into {\ML}, and we limit the discussion scope to simple and representative models for ease of discussions.

\vspace{1mm}
\noindent\textbf{Polynomial Regression (Poly)~\cite{hastie2009elements}} uses basis functions to capture nonlinear relationships between variables, replacing linear terms in linear regression. 
%
%
We treat each term from the cost models in Figure~\ref{fig:example} as basis functions, and include a constant term for each operation type to account for CPU time consumption. 
The regression model for the cost function can be formulated as:
\begin{equation}
\label{equ:poly_cost}
y_{cost}=\sum \beta_{i} x_i
\end{equation}
Here, $\beta_i$ represents the coefficients to be learned, and $x_i$ denotes the basis functions derived from theoretical models. 


\vspace{1mm}
\noindent\textbf{Tree Ensembles}
(Trees)~\cite{chen2015xgboost,marcus2019plan} are a type of ML model that combines multiple decision trees to enhance accuracy and robustness. 
The model can mitigate the impact of individual tree biases and errors, resulting in more accurate and stable predictions. 
In this paper, we primarily focus on gradient-boosted trees due to its widely adoption~\cite{wang2020we,dutt2019selectivity}.
Tree ensembles offer the advantage of bypassing manual basis function design and directly incorporating cost function elements. They detect feature relationships automatically and have strong fitting capability. However, they are prone to overfitting and outlier sensitivity due to uncertainty, and their extrapolation ability is weaker than polynomial regression.

During training and inference, we input the influential factors as independent features into the tree ensembles, including $N$, $T$, $M_b$, $M_c$, $v$, $r$, $w$, and $q$. The average latency is used as the label, and the ensemble trees automatically capture the relationships between these features.

\vspace{1mm}
\noindent\textbf{Neural Networks.}
We also explore more complex and advanced models, such as neural networks (NN), which consist of interconnected nodes organized into layers for data processing and analysis. However, these complex models typically demand a larger amount of training data~\cite{juba2019precision,hu2021model,ying2019overview}, resulting in increased sampling time. 
To verify the expectation, while considering previous experiences~\cite{marcus2019plan,sun13end} and conducting experimental model selection, we evaluate a standard NN model with four fully connected layers. This model requires three times the number of samples compared to the other two models in order to achieve similar optimization outcomes (see more details in Figure~\ref{fig:hours} in Section~\ref{sec:eva}).

\section{Evaluation}
\label{sec:eva}
\begin{figure*}
   \hspace{-8mm}
    \begin{subfigure}{0.71\textwidth}
        \centering
        \includegraphics[width=\textwidth]{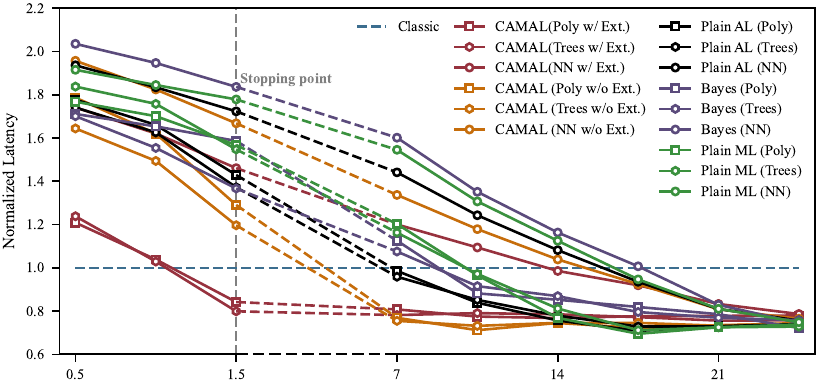}
            \vspace{-5mm}
        \caption{Sampling hours of different strategies}
        \label{fig:hours}
    \end{subfigure}
\raisebox{-1.5mm}{
    \begin{subfigure}{0.22\textwidth}
        \centering
\footnotesize
  \tabcolsep=0.9mm
  \renewcommand{\arraystretch}{0.8}
  \begin{tabular}{l||lll} 
  \textbf{Methods} & \textbf{Mean} & \textbf{90th} 
        \\ \hline \hline
        \ML (Poly) & 0.11 & 0.14\\
        \ML (Trees) & 0.10 & 0.14\\
        \ML (NN) & 0.19 & 0.22\\
        Plain AL (Poly) & 0.19 & 0.23\\
        Plain AL (Trees) & 0.18 & 0.24\\
        Plain AL (NN) & 0.22 & 0.29\\
        Bayes (Poly) & 0.22 & 0.26\\
        Bayes (Trees) & 0.19 & 0.25\\
        Bayes (NN) & 0.25 & 0.28\\
        Plain ML (Poly) & 0.18 & 0.24\\
        Plain ML (Trees) & 0.19 & 0.25\\
        Plain ML (NN) & 0.24 & 0.29\\
        Classic & 0.13 & 0.18\\
        Classic (Cache) & 0.12 & 0.17\\
        Monkey & 0.14 & 0.20\\
        \end{tabular}
        \caption{Latency/Op (ms) }
        \label{tab:overall_lat}
        \vspace{0.1cm} 
    \end{subfigure}}

\hspace{-2mm}
\raisebox{0mm}{
    \begin{subfigure}{0.2\textwidth}
        \centering
\footnotesize
  \tabcolsep=0.9mm
    \renewcommand{\arraystretch}{0.8}
  \begin{tabular}{l||lll} 
  \textbf{Methods} & \textbf{Mean}  
        \\ \hline \hline
{\ML} (Poly) & 6.2  \\ 
{\ML} (Trees) & 4.5 \\ 
{\ML} (NN) & 32.8 \\ 
Plain AL (Poly) & 15.2 \\ 
Plain Al (Trees) & 15.5 \\ 
Plain AL (NN) & 34.7 \\ 
Bayes (Poly) & 16.0 \\ 
Bayes (Trees) & 14.4 \\ 
Bayes (NN) & 34.2 \\ 
Plain ML (Poly) & 21.3 \\ 
Plain ML (Trees) & 18.6 \\ 
Plain ML (NN) & 34.9 \\   
Classic & 16.2 \\ 
Classic (Cache) & 13.4 \\ 
Monkey & 24.3
        \end{tabular}
        \caption{I/Os per operation}
        \label{tab:overall_io}
    \end{subfigure}}
    \hspace{4mm}
     \begin{subfigure}{0.75\textwidth}
        \centering
            \includegraphics[width=\textwidth]{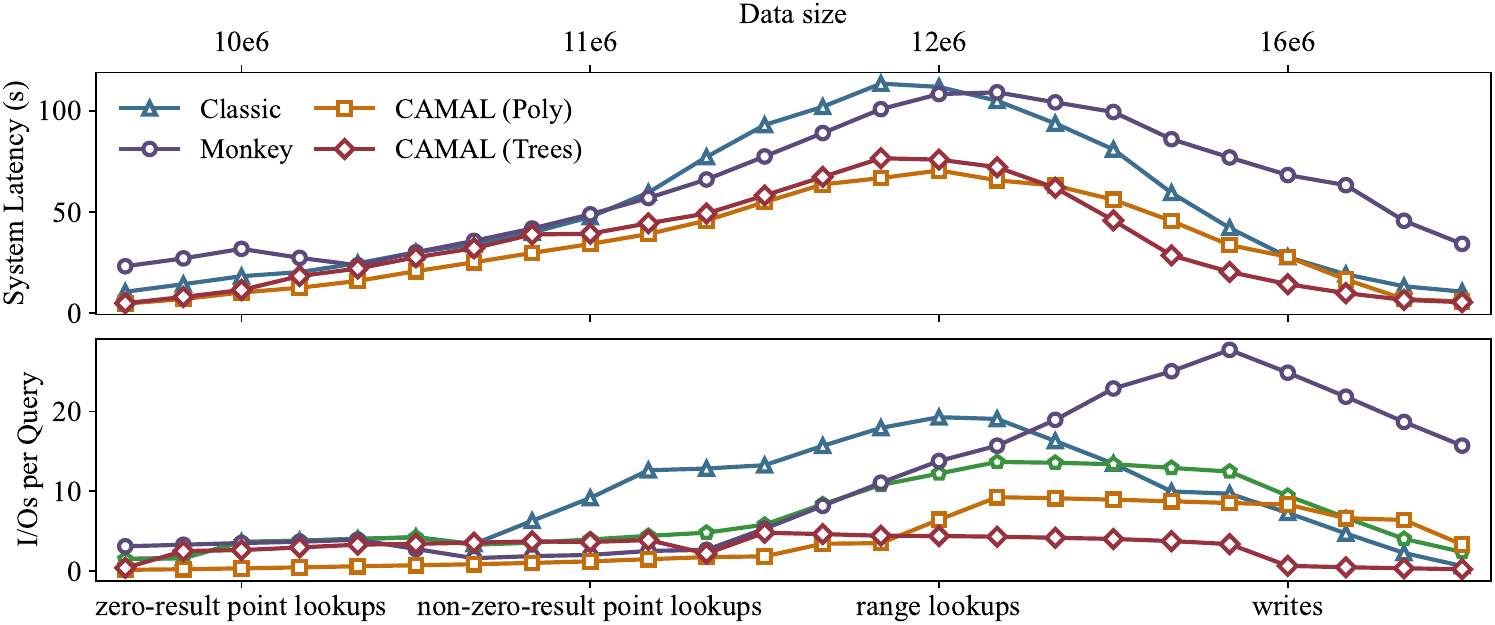}
            \vspace{-1mm} 
        \caption{Latency and I/O comparison on dynamic 
 test workloads}
        \label{fig:uniform_shift}
    \end{subfigure}
    \caption{
Classic methods relying solely on theoretical cost models necessitate no training samples but demonstrate limited optimization. Conversely, other ML or AL methods that employ plain sampling strategies achieve high performance but entail relatively extensive sampling costs. {\ML} is the only complexity-analysis driven ML-aided framework that achieves high performance while significantly reducing sampling costs.}
    \label{fig:overall}
\end{figure*}

We integrated {\ML} into RocksDB and conduct systematic experimental evaluations to demonstrate its effectiveness. 
\vspace{-3mm}

\subsection{Experimental Setup}
\noindent\textbf{Hardware.} \revise{Our experiments are done on a server with Intel(R) Core(TM) i9-13900K processors, 128GB DDR5 main memory,
2TB NVMe SSD, a default page size of 4 KB, and running 64-bit Ubuntu 20.04.4 LTS on an ext4 partition.}


\vspace{1mm}
{
\noindent\textbf{Database Setup.} We employ RocksDB~\cite{rocksdb}, a widely-used LSM tree-based storage system. Consistent with Endure~\cite{huynh2022endure}, we evaluate the steady-state performance of databases by initializing each instance with 10 million unique key-value pairs, where each pair is 1 KB in size, unless otherwise specified. Every key-value entry includes a randomly selected 16-bit key, and the rest of the bits are filled with randomly generated values. 
}

We allocate a varying number of bits per element for the Bloom filters at each level based on Monkey~\cite{dayan2017monkey}. The total memory budget for this is set by default to 16MB. In our default experimental setup, we focus on adjusting only the size ratio and the memory allocated to the write buffer and Bloom filter as tunable parameters. Additionally, the leveling compaction policy is the only one we consider in this default setting, unless explicitly stated otherwise. This implies that other compaction policies and parameter adjustments are explored only in specific, later experiments. \revise{Following Endure~\cite{huynh2022endure}, to obtain an accurate count of block accesses we enable direct I/Os for both queries and compaction and disable the compression.}

{
}

\vspace{1mm}
{
\noindent\textbf{Workloads.} Our evaluation is divided into static and dynamic modes. The queries are generated based on both the 15 standard workloads detailed in Table~\ref{tab:train_workload} and the 24 shifting workloads outlined in Table~\ref{tab:test_workload}. The shifting workloads follow the principle of progressively transitioning weights among different operations, enabling {\ML} to gradually adapt the parameters and include a broad spectrum of application scenarios. In the static mode, we assess both standard and shifting workloads, whereas in the dynamic mode, we focus solely on the shifting workloads. By default, our evaluations are conducted in static mode, unless otherwise specified. In the dynamic mode, to ensure adequate compactions for dynamic LSM's gradual parameter adjustment, we continuously insert additional entries during the intervals of workload queries.
}

We assess the latency and I/O times of a series of 500,000 queries, reporting the average performance across all workloads, unless specified differently. All range queries are set up with minimal selectivity, functioning as short-range queries that generally access between zero to two blocks per level on average.

{
Regarding data distribution, we employ both uniform and Zipfian query distributions, utilizing YCSB~\cite{cooper2010benchmarking} for implementation. We also modify the skew coefficient of the Zipfian distribution, adjusting it within a range from 0 to 0.99. This distribution is applied to both the data being ingested and the queries.}

%
%
\begin{table}[tb!]
\footnotesize
\centering
\begin{center}
  \caption{Operation percentages in 24 test workloads.}  
  \label{tab:test_workload}
   \renewcommand{\arraystretch}{0.9}
  \tabcolsep=1mm
  \begin{tabular}{l||llllllllllllllllllllllll} 
&\multicolumn{6}{>{\centering\arraybackslash}m{2.2cm}}{\textbf{zero-result point lookups}}&\multicolumn{6}{{>{\centering\arraybackslash}m{2.2cm}}}{\textbf{non-zero-result point lookups}}&\multicolumn{6}{c}{\textbf{range lookups}}&\multicolumn{6}{c}{\textbf{writes}}
   \\ \hline \hline
   $v$& 60&75& 91 & 75&60&45&30&15&3&5&5&5&5&5&3&5&5&5&5&5&3&15&30&45\rule{0pt}{0.3cm}\\
   $r$& 5&5& 3 & 15&30&45&60&75&91&75&60&45&30&15&3&5&5&5&5&5&3&5&5&5\rule{0pt}{0.3cm}\\
   $q$& 5&5& 3 & 5&5&5&5&5&3&15&30&45&60&75&91&75&60&45&30&15&3&5&5&5\rule{0pt}{0.3cm}\\
   $w$& 30&15& 3 & 5&5&5&5&5&3&5&5&5&5&5&3&15&30&45&60&75&91&75&60&45\rule{0pt}{0.3cm}\\
\end{tabular}
\end{center}
\end{table}

\vspace{1mm}
\noindent\textbf{{\ML} Setup.}
In our evaluation, the distinction between using (w/ Ext.) or not using (w/o Ext.) the extrapolation strategy is made only in the first experiment. In all other experiments, we consistently apply the extrapolation strategy. We set the factor $k$ in the extrapolation strategy to 10, meaning we train with only 1/10th of the testing data size and memory budget. This choice is justified as this scaling factor strikes a balance between minimizing sampling costs and maximizing optimization performance, a concept further evidenced in subsequent experiments. Thus, unless otherwise specified, {\ML} is trained with the scaled-down setting and extrapolation strategy and tested under full settings.

As for the ML models, we implement polynomial regression using the least square method~\cite{harris2020array} ({\ML}~(Poly)), tree ensembles with XGBoost~\cite{chen2015xgboost} ({\ML}~(Trees)) and neural network with PyTorch~\cite{paszke2019pytorch} ({\ML}~(NN)). Once 
training samples are ready, all models can be trained in 5 seconds and traversed in 200 milliseconds to search the \revise{desired}  parameters, with a space overhead of under 200KB. This is negligible compared to the time saved through optimization and the total memory budget. 

\vspace{1mm}
\noindent\textbf{Implementation optimizations.} We offer three applications with {\ML} to better integrate our system into a real key-value store for practical purposes. First, we incorporate block cache memory allocation as an input to the ML models to optimize cache strategy. Second, we consider data distribution beyond a default uniform setting - mainly Zipfian distribution~\cite{cooper2010benchmarking,lu2017wisckey,dai2020wisckey,chatterjee2021cosine}. Here we discuss three choices for incorporating different levels of distribution knowledge: (a) When the data distribution is unknown during runtime, we simply train the models using uniform data and test them with arbitrary data. (b) If the coefficient that represents the data distribution during runtime can be determined, we train the model using the same distribution. (c) When multiple potential coefficients exist for a data distribution during runtime, we integrate the coefficient as an input feature within the ML model. Third, we address workload uncertainty, which is the inconsistency between observed and expected operation proportions, a challenge initially addressed by Endure~\cite{huynh2022endure}. We adopt the hyperparameter $\rho$, representing the subjective uncertainty region anticipated in testing workloads defined by KL-divergence distance~\cite{huynh2022endure}. Our solution is statistically based -- while maintaining the training process, during testing, we randomly sample several workloads within a region of size $\rho$ and identify settings with the lowest average latency across these workloads as the \revise{desired} solution.

\vspace{1mm}
\noindent\textbf{Baselines.}  
We compare the performance of {\ML} against five LSM-tree tuning techniques: (1) well-tuned RocksDB with Monkey~\cite{dayan2017monkey}, (2) classic tuning implemented in Endure~\cite{huynh2022endure}, (3) plain ML methods that ignore sampling improvement techniques, (4) plain active learning without considering theoretical-driven decoupled sampling, and (5) Bayesian Optimization as sampling, which has been adopted in tuning relational databases~\cite{van2017automatic,duan2009tuning}. The first two techniques represent industry and research practices, while the latter three serve as baselines to show the improvement of {\ML} over classic ML sampling techniques. Our evaluation of well-tuned RocksDB incorporates the memory allocation strategy for the bloom filter from Monkey~\cite{dayan2017monkey} and minor fixes from Spooky~\cite{dayan2022spooky}. We also include a set of reasonable parameters based on experience and experimental settings, such as using leveling compaction with a size ratio of 10, $M_f$ set to $10N$ bits (i.e., 10 bits per key), and the remaining budget allocated to $M_b$. For classic tuning, we apply the nominal tuning approach used in Endure~\cite{huynh2022endure}, which minimizes the classic I/O cost model using the Sequential Least Squares Programming optimizer~\cite{virtanen2020scipy}. For plain ML methods, we employ grid search to partition the sampling space according to the sampling budget. As for plain active learning, we choose a set of samples per round similar to {\ML}. 
We also employ a prevalent implementation of Bayesian Optimization with Gaussian processes~\cite{bo2014}. 


\begin{figure*}[t]
  \centering
  \raisebox{4mm}{
  \hspace{-2mm}
  \begin{subfigure}[b]{0.23\textwidth}
  \tabcolsep=1.1mm
    \renewcommand{\arraystretch}{1}
    \tiny
     \begin{tabular}{p{1cm}||lll} 
      \textbf{$N$} & $1e6$ & $2e6$ &$1e7$   
            \\ \hline \hline
         \textbf{Classic}& 1& 1  & 1\\ 
     \textbf{Poly}& 0.84 & 0.82  & 0.83\\ 
     \textbf{Trees} & 0.82 & 0.81 & 0.82\\ 
    \end{tabular}

    \vspace{1mm}
        \renewcommand{\arraystretch}{1}
     \begin{tabular}{p{1cm}||lll} 
      \textbf{$M  (MB)$} & $16$ & $32$ &$64$   
            \\ \hline \hline
    \textbf{Classic}& 1& 1  & 1\\ 
     \textbf{Poly}& 0.84 & 0.82  & 0.86\\ 
     \textbf{Trees} & 0.82 & 0.83 & 0.84\\ 
    \end{tabular}
       \captionsetup{font=footnotesize,skip=4pt}
       \vspace{2mm}
            \caption{\footnotesize Normalized Latency}
        \label{tab:eff}
  \end{subfigure}
  }
  \raisebox{4mm}{
  \begin{subfigure}[b]{0.24\textwidth}
            \includegraphics[width=\textwidth]{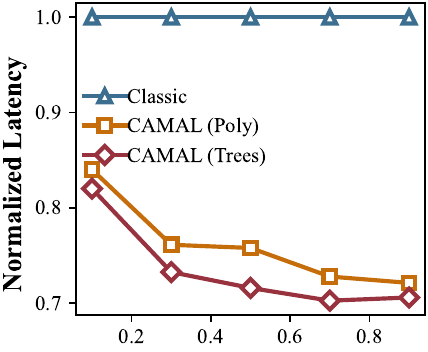}
   \captionsetup{font=footnotesize,skip=0pt,margin={3pt,0pt}}
    \caption{\footnotesize Skew Coefficient}
    \label{fig:comp_zipf}
  \end{subfigure}
    \hspace{0.5mm}
  \begin{subfigure}[b]{0.24\textwidth}
    \includegraphics[width=\textwidth]{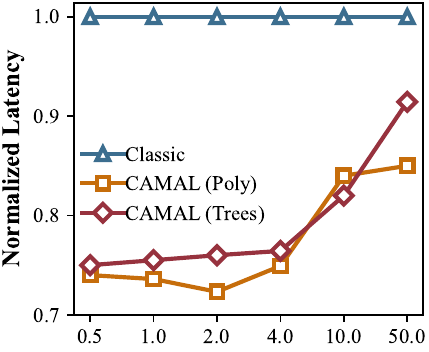}
\captionsetup{font={scriptsize},skip=2pt,margin={3pt,0pt}}
\caption{\footnotesize $k$ in extrapolation}
    \label{fig:comp_extrap_n}
  \end{subfigure}
      \hspace{0.5mm}
  \begin{subfigure}[b]{0.24\textwidth}
    \includegraphics[width=\textwidth]{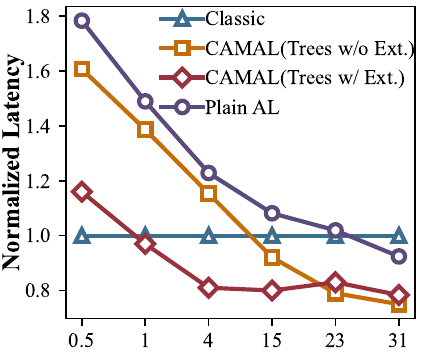}
\captionsetup{font={scriptsize},skip=-1pt,margin={-5pt,-6pt}}
    \caption{\footnotesize Large data sampling hours}
    \label{fig:larger}
  \end{subfigure}}

\vspace{-2mm}
  \raisebox{2mm}{
    \hspace{-3mm}
  \begin{subfigure}[b]{0.24\textwidth}
    \includegraphics[width=\textwidth]{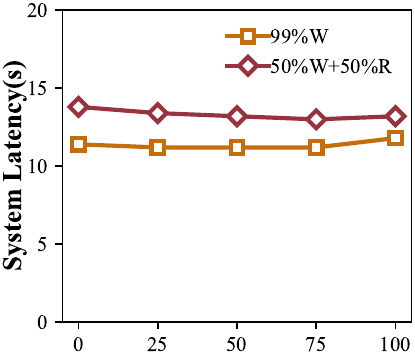}
   \captionsetup{font=footnotesize,skip=4pt,margin={3pt,0pt}}
    \caption{\footnotesize \% deletes in workload}
    \label{fig:del}
  \end{subfigure}
      \hspace{1mm}
\begin{subfigure}[b]{0.24\textwidth}
    \includegraphics[width=\textwidth]{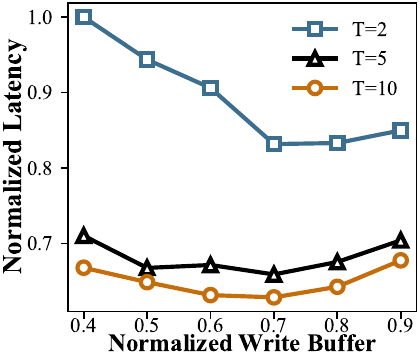}
   \captionsetup{font=footnotesize,skip=-1pt,margin={-2pt,0pt}}
  \caption{\footnotesize Independence example}
  \label{fig:inde}
  \end{subfigure}
        \hspace{0.5mm}
 \begin{subfigure}[b]{0.24\linewidth}
    \includegraphics[width=\textwidth]{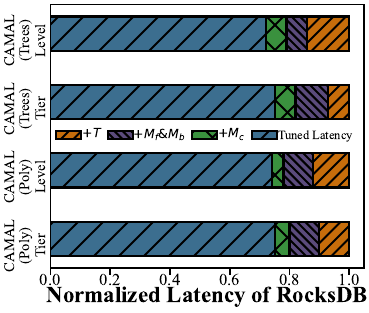}
   \captionsetup{font=footnotesize,skip=-1pt,margin={-3pt,-5pt}}
    \caption{\footnotesize Parameters Breakdown}
    \label{fig:split_para}
  \end{subfigure}
          \hspace{0.8mm}
  \begin{subfigure}[b]{0.24\linewidth}
    \includegraphics[width=\textwidth]{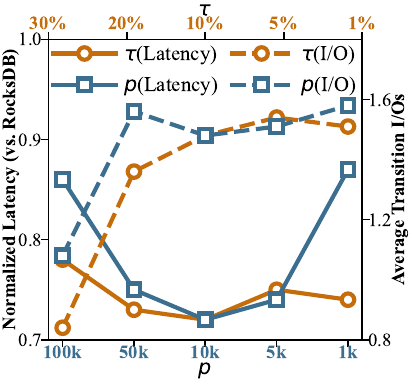}
   \captionsetup{font=footnotesize,skip=-1pt,margin={-2pt,0pt}}
    \caption{\footnotesize Sensitivity of $\tau$ and $p$}
    \label{fig:detect}
  \end{subfigure}}
  
  \caption{\revise{ The robustness of {\ML}.}}
  \label{fig:comp}
\end{figure*}

\subsection{Overall System Performance}
\label{subsec:system}

{
\noindent\textbf{Decoupled active learning in {\ML} significantly reduces sampling costs. } As illustrated in Figure~\ref{fig:hours}, even without the extrapolation strategy, our sampling process using {\ML} with polynomial regression and tree ensembles can be finished in just 7 hours. In contrast, alternative methods typically require over 20 hours to reach a similar level of performance. The faster identification of the \revise{desired} solution by {\ML} is attributed to its unique combination of complexity-analysis driven sampling and ML-aided sampling. Plain ML methods often fail to locate a suitable solution within a limited sampling budget, as their grid search approach divides the sampling space uniformly. Bayesian optimization, while more efficient than grid search, still falls short in some aspects. It better approximates \revise{desired} settings by using prior information within the current workload, but it requires additional iterations for exploration. This is because it tends to explore each workload in an instance-optimized LSM-tree independently, without utilizing information from other workloads. Additionally, its random initialization for each training workload leads to more exploration iterations to find the \revise{desired} solution. Plain active learning, despite its ability to leverage information across different workloads, is only marginally more effective than Bayesian optimization due to its random initialization approach. This limitation prevents it from outperforming the more integrated and efficient approach offered by {\ML} with decoupled active learning.
}

\vspace{1mm}
\noindent\textbf
{The extrapolation strategy in {\ML} can greatly reduce sampling costs when used with appropriate settings.} Figure~\ref{fig:hours} demonstrates that, with extrapolation, {\ML} can save around 80\% of sampling time (1.5 hour vs. 7 hours). \revise{The time includes generating three samples for each group of parameters in each training workload, which means the average time to generate a sample is about 40 seconds. Note that the model training time is negligible compared to the sampling time.} However, there are important considerations for achieving this efficiency: (1) Scaling Factor Range: As Figure~\ref{fig:comp_extrap_n} indicates, while optimization performance remains relatively stable up to a scaling factor of 10, it drops significantly when the factor exceeds about 50. This decline is due to two main reasons: firstly, as the training database size diminishes, the uncertainty in system latency becomes more pronounced. Secondly, other hardware factors, such as garbage collection (as described in Spooky~\cite{dayan2022spooky}), also affect system latency. (2) Performance Gap with Full-Size Training: When using extrapolation, {\ML} incurs around 5\% more system latency compared to training with a full-size database instance. Adding more scaled-down training data does not bridge this gap. This discrepancy is attributed to both the increased uncertainty and the aforementioned hardware factors. In conclusion, while the extrapolation strategy can significantly optimize performance, its effectiveness is subject to other engineering factors that must be carefully managed.

\vspace{1mm}
\noindent\textbf
{The choice of embedded ML models significantly influences {\ML}'s performance.} As shown in Figure~\ref{fig:uniform_shift}, when {\ML} incorporates polynomial regression and tree ensembles, we observe comparable optimization levels after 7 hours of sampling. However, using a neural network as the base model leads to subpar performance at the same stopping point. Although employing extrapolation reduces sampling costs for neural networks, their performance still lags significantly behind the other two ML models. Similar patterns are seen when neural networks are combined with other sampling methods like Bayesian optimization and grid search. This trend could be attributed to the fact that deep neural networks typically require more extensive training data to generalize effectively and avoid overfitting~\cite{ying2019overview,hu2021model}. This necessity for extensive data underpins our preference for polynomial regression and tree ensembles as primary ML models in {\ML}. As discussed in Section~\ref{sec:ml}, polynomial regression and tree ensembles have a moderate complexity level compared to neural networks, making them more suitable for scenarios with fewer training samples.




\vspace{1mm}
\noindent\textbf{{\ML} consistently achieves better end-to-end performance than traditional tuning techniques.} Figure~\ref{fig:hours} shows that {\ML} with polynomial regression and tree ensembles consistently achieves lower latency (per operation) compared to classic tuning methods across various workloads, with an average improvement of 16\% $\sim$ 18\% (0.10 $\sim$ 0.11 ms vs. 0.13 ms). This can be attributed to {\ML}'s ability to capture the true relationship between parameters and the actual system conditions. In contrast, classic methods may provide biased estimates as they are purely based on I/O cost and neglect CPU time. As a result, it may still lead to high I/O cost, as shown in the range lookups-heavy workloads in Figure~\ref{fig:uniform_shift}. Further investigation reveals that theoretical models do not fit well with the {\it seek} and {\it next} operations in RocksDB, leading to unexpected performance. Monkey offers a universal solution for all workloads, leading to less variance but higher latency than classic tuning methods. Compared to Monkey, {\ML} significantly reduces latency by up to 8x (e.g., for some write-heavy workloads in Figure~\ref{fig:uniform_shift}), thanks to its ability to adapt to the unique operation ratios in each workload.

Although {\ML}'s learning objective does not specifically target I/O cost, it generally leads to a significant reduction in I/O cost compared to other methods in most cases. When compared to classic tuning, {\ML} reduces I/O cost by up to 72\% (4.5 vs. 16.2). This happens because classic tuning, despite aiming to minimize I/O cost, may not always achieve lower I/O cost due to inaccurate models. In particular, optimizing the theoretical complexity-based cost model results in sub-optimal solutions for system I/Os, such as the range lookup-heavy and write-heavy workloads in Figure~\ref{fig:uniform_shift}. The rationale behind this is that the real database does not perform exactly as predicted by the theoretical model. Monkey provides a default setting that is relatively stable for all workloads, resulting in higher I/Os compared to other methods, especially in write workloads. Since {\ML} has lower end-to-end latency, it implies that the sum of the I/O time and the CPU time is lower, which typically implies fewer I/Os. Our insight is that low latency is a sufficient condition for a low I/O cost, but low I/O cost does not necessarily implies a low latency due to potentially higher CPU costs. This explains the rationale for using latency as a learning objective.

\vspace{1mm}
{\noindent\textbf{{\ML} exhibits superior performance on the dynamic LSM-tree.} The trends illustrated in Figure~\ref{fig:uniform_shift} are derived from the dynamic LSM-tree under progressively changing workloads, where {\ML} distinctly outperforms competing methods. \revise{To further clarify the setup in Figure~\ref{fig:uniform_shift}, each point of x-axis represents a workload in Table~\ref{tab:test_workload}, in which the ratio of operations shifts gradually. Internal the workloads we adopt the detection methods and reconfigure the dynamic LSM-tree which has been elaborated in Section~\ref{sec:dynamic}. As the workloads include writes, the data size will gradually increase. So we also employ the extrapolation strategy elaborated in Section~\ref{sec:extrap}. We report the system latency and average I/Os of each workload.} The trajectory of {\ML}  demonstrates that optimization performance progressively adapts as workloads evolve. This points to the dynamic LSM-tree’s ability to adjust parameters, thereby surpassing the performance of a statically configured LSM-tree. The primary reason for this is that while Monkey and RocksDB can deliver a generally effective configuration for a range of workloads—maintaining stable performance—they lack the capability to fine-tune settings for specific workload scenarios. In contrast, the dynamic LSM-tree, guided by ML, can incrementally modify parameters to align closely with the theoretically optimal configuration. 

  \subsection{Parameter Sensitivity Study}
\noindent\textbf{The impact of each tunable parameter.}
\revise{As shown in Figure~\ref{fig:split_para}, we break down the impact of each parameter on performance compared with well-tuned RocksDB. Initially, we configure $M_f$ as $10N$, set $M_b$ to $M-M_f$, and $M_c$ to 0, while focusing on tuning $T$. The label " +$T$" indicates tuning only $T$; " +$M_f \& M_b$" involves balancing between $M_f$ and $M_b$; and " +$M_c$" additionally includes tuning $M_c$. The results reveal that adjusting $T$ leads to significant performance improvements (reducing normalized latency from 1 to 0.88/0.86 in two models with leveling), while tuning $M_c$ also provides substantial benefits. This suggests that our feature optimization approach in {\ML} is in a reasonable order of importance.}


\revise{When further examining the optimization effects based on the compaction policy, as shown in Figure~\ref{fig:split_para}, we find that both leveling and tiering compaction policies achieve comparable effectiveness after tuning with {\ML}. However, combining these policies can result in better optimization. This finding highlights the potential of selecting an appropriate compaction policy to maximize optimization outcomes. Therefore, in Section~\ref{subsec:case}, we further discuss how to extend our model to incorporate the number of runs in an LSM-tree level as another tunable parameter for system optimization. }

\revise{We also experimentally verify the parameters' independence. We vary $T$ under specified workloads and then explore the sampling space for memory allocation to the write buffer and Bloom filter. As illustrated in Figure~\ref{fig:inde}, with varying size ratios, the tuned latency consistently approaches the practical desired when the write buffer occupies 60\% to 70\% of the memory budget. This supports the approach of tuning $T$ first, followed by memory allocation.  }

\vspace{1mm}
\noindent\revise{{\bf The impact of data scales.} To demonstrate the scalability of {\ML}, we explore settings with larger data volumes and increased memory budgets. Specifically, we expand the number of entries to 50 million and the memory budget to 80MB proportionally. As illustrated in Figure~\ref{fig:larger}, {\ML} maintains approximately the same optimization performance as seen with the 10MB setting, both at 4 and 23 hours, with and without extrapolation. In contrast, plain active learning (AL) still underperforms relative to {\ML}, achieving only a 5\% reduction in latency over 31 sampling hours. This change in sampling hours is primarily due to the increase in the sampling cost per sample compared to the original settings.}

\vspace{1mm}
\noindent\revise{\textbf{The impact of delete workload.} To investigate the generalizability of {\ML} to a broader range of operations, we evaluate its performance in static scenarios with not only updates but also deletes. We vary the ratio of deletes within the writes across two workloads. As illustrated in Figure~\ref{fig:del}, the tuned system latency remains nearly unchanged in both scenarios. This is because, unless specifically targeted by optimization strategies, deletes are handled in much the same way as inserts and updates in LSM-trees.}

\vspace{1mm}
\noindent\textbf{\revise{The sensitivity of parameters for dynamic mode.}} \revise{We have detailed the detection method in Section \ref{sec:dynamic}. In our default setting, we set $\tau=10\%$ and $p=10k$ ($k$ stands for thousand), where we split the entire dynamic workload into sub-workloads of $10k$ operations each. For example, if the last reconfiguration was triggered when $w=30\%$, and in the current sub-workload $w=15\%$, we start the reconfiguration; if $w=25\%$, we keep the configuration unchanged and move on to the next sub-workload.}

\revise{We examine the sensitivity of $\tau$ and $p$ in Figure \ref{fig:detect}. Decreasing $p$ below $50k$ improves the post-tuning system latency, achieving 72\% $\sim$ 75\% normalized latency, but when $p$ is smaller than $1k$ operations degrades performance as it cannot accurately inspect the true percentage of the randomly generated workloads. For $\tau$, the post-tuning system latency becomes stable when the threshold is below 20\%. Further, decreasing $p$ and $\tau$ will both increase the transition I/Os, which is the I/O cost incurred during reconfiguration, when reconfiguration becomes more frequent. However, the transition I/Os will stabilize when $\tau<20\%$ and $p<50k$, and compared to the post-tuning system I/Os, they are relatively small. This is because the lazy transition strategy we adopted ensures that the LSM-tree structure only changes during natural compaction, even though with more frequent reconfigurations. We also note that when the workload change gradually, which is often the case in practice, the tuned LSM-tree does not need to change dramatically. }

\revise{In summary, moderate values of $\tau$ (5\% to 20\%) and $p$ ($5k$ to $50k$) would provide a range where the system performance after tuning is less sensitive to the parameters.}


\begin{figure}[t]
  \centering
  \hspace{-2mm}
  \begin{subfigure}[b]{0.24\textwidth}
            \includegraphics[width=\textwidth]{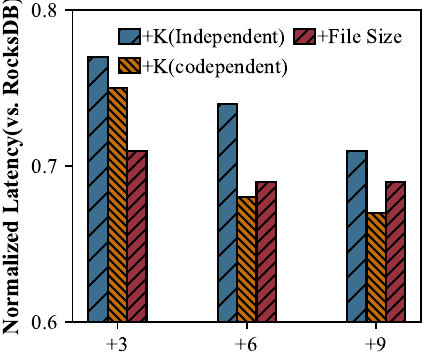}
   \captionsetup{font=footnotesize,skip=1pt,margin={-2pt,-2pt}}
    \caption{\scriptsize Samples of new parameters}
    \label{fig:new_para}
  \end{subfigure}
  \hspace{0.8mm}
  \begin{subfigure}[b]{0.24\textwidth}
    \includegraphics[width=\textwidth]{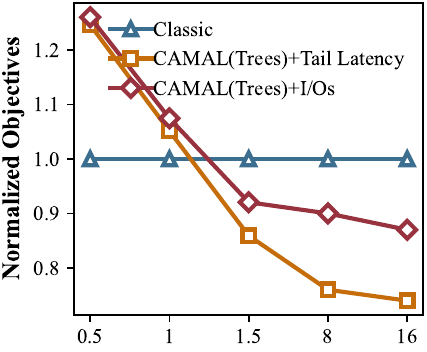}
   \captionsetup{font=footnotesize,skip=1pt,margin={-2pt,-2pt}}
    \caption{\scriptsize Objectives sampling hours}
    \label{fig:obj}
  \end{subfigure}
    \hspace{0.8mm}
  \begin{subfigure}[b]{0.24\textwidth}
    \includegraphics[width=\textwidth]{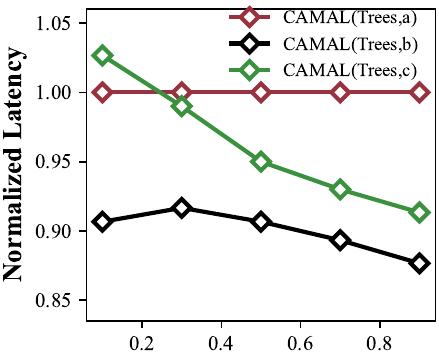}
   \captionsetup{font=footnotesize,skip=1pt,margin={-2pt,-2pt}}
    \caption{\scriptsize Strategy vs. Skewness}
    \label{fig:comp_dist_stra}
  \end{subfigure}
  \hspace{0.8mm}
    \begin{subfigure}[b]{0.235\textwidth}
    \includegraphics[width=\textwidth]{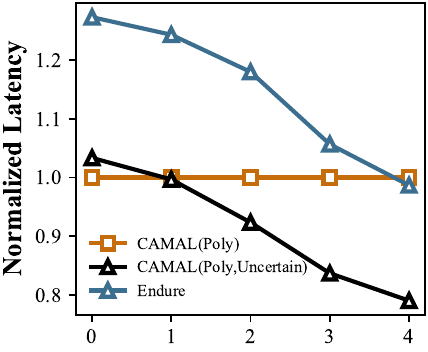}
   \captionsetup{font=footnotesize,skip=1pt,margin={-2pt,-2pt}}
        \caption{\scriptsize Uncertainty Region}
    \label{fig:comp_uncert}
  \end{subfigure}
\caption{{\ML} demonstrates versatility with potential extensions to new parameters, alternative objectives, and varied data distributions.}
\end{figure}

\subsection{Model Extension and Discussion}
\label{subsec:case}
\noindent\revise{In this section, we broaden the application across various scenarios to validate the generalization capabilities of {\ML} and to explore its potential further. This involves evaluating the model against diverse data distributions, incorporating a wider array of tunable parameters, exploring the use of alternative learning objectives.}

\vspace{1mm}
\noindent\textbf{Performance under different data distributions.} As shown in Figure~\ref{fig:comp_zipf}, compared with traditional tuning methods, {\ML} can accurately capture the actual distribution of accessed data blocks, which further allows more precise configurations. In addition, with an increase in the skew coefficient, certain blocks experience higher access frequency, resulting in improved efficiency of the block cache. Consequently, {\ML} facilitates more rational memory allocation, wherein the block cache receives a greater allocation of memory when the data distribution exhibits higher skewness.

As previously discussed, we have proposed three training options that can be employed if information regarding the data distribution is available. As shown in Figure~\ref{fig:comp_dist_stra}, strategies (b) and (c) can each enhance the optimization of strategy (a) by approximately up to 15\% when skewness increases. This improvement can be attributed to more effective cache allocation based on the known data distribution, as cache sensitivity to data distribution is demonstrated in many works~\cite{yang2020leaper,berthet2017approximation}.

\vspace{1mm}
\noindent\textbf{Extend to uncertain workloads.} 
We follow the setting in Endure~\cite{huynh2022endure}, where the expected uncertainty regains are the same as the observed regions. As shown in Figure~\ref{fig:comp_uncert}, the original {\ML}, which does not take uncertainty into account, can outperform Endure for reasonable uncertain ranges. This is because {\ML}'s tuning performance is significantly better than the classic tuning adopted in Endure, so small uncertainty does not have enough impact to cause change. If we take workload uncertainty into account, the effect becomes more pronounced compared to the original settings. This is because we choose the settings that lead to the samples with the lowest average latency in the uncertainty region based on accurate estimations.

\vspace{1mm}
\noindent\revise{
\textbf{The possibility of involving more parameters based on group-wise sampling.} In addition to the tunable parameters explored in previous experiments, numerous specific or practical parameters exist in various LSM-tree systems. 
Here we discuss one possible solution to extend {\ML} for considering these parameters. Through theoretical and experimental analysis, 
we observe that some parameters are strongly dependent. Therefore, we can group strongly dependent parameters together, and then sample and fine-tune the parameters within one group before proceeding to another group. For instance, let us consider two more parameters -- the number of runs at each level ($K$), and the file size of each SST. We will then select to group $K$ and $T$, as they are strongly dependent~\cite{mo2023learning,huynh2023flexibility, liu2024structural}. Meanwhile, the SST file size does not have an explicit correlation with other parameters we chose. 
Within each group, we would sample the parameters together. For instance, our solution groups $T$ and $K$ in a 2-dimensional sampling space. In the active learning cycle, we first calculate the theoretical optimal pair $(T^*, K^*)$. Then, based on the given sampling budget, we sample these two parameters in a 2-dimensional neighborhood. The subsequent steps follow Algorithm\ref{alg:decopuled}, except that we obtain the tuned combination $(T', K')$ using ML models, instead of just tuning $T'$. For the file size, as it has not been shown to have an explicit correlation with our existing parameters ($T$, $K$, and memory-related parameters), we sample it after the other parameters. Following our decoupled active learning paradigm, we initiate with a practically reasonable file size (e.g., the default file size in RocksDB) as the sampling center and then obtain samples from its neighbors.}

\revise{As shown in Figure~\ref{fig:new_para}, we evaluate two strategies for incorporating the parameter $K$. When sampling $K$ and $T$ separately, which means we first obtain a tuned $T^*$ and then sample $K$ based on $T^*$, the optimization performance is noticeably inferior compared to when they are sampled together. This is because $K$ and $T$ are closely correlated, and sampling them independently is likely to result in a local optimum for a narrow range of sampled $T$. For co-dependent sampling, we also need to add extra samples to achieve better performance, as the correlation requires more data to be accurately learned. On the other hand, we sample the file size independently since we have not identified any explicit correlation between it and other tunable parameters in this study. The optimization impact is not as pronounced as with other parameters, because it primarily affects space amplification rather than I/Os or latency~\cite{dayan2022spooky}.}

\revise{We acknowledge the complexity inherent in tuning LSM-tree parameters, as they can be interdependent, slightly dependent, or only show dependency under certain queries. This complexity guides us to thoughtfully categorize parameters within the {\ML} framework for practical examination: (1) The first category consists of parameters that are crucial to the LSM-tree's structure, notably $T$ and $K$. (2) The second category encompasses memory-related parameters, such as the write buffer, Bloom filter, and block cache. (3) The third category is dedicated to detailed merge policies, like file size, along with system-specific parameters, such as sub-compaction and compaction granularity. This categorization reflects observed correlations: parameters within each category tend to have strong or inherent correlations, whereas those across different categories generally display relatively implicit connections. Although this classification might not align perfectly with every theoretical consideration, the flexible nature of ML models to generalize could help mitigate the gaps in terms of performance. Exploring the integration of additional parameters further, particularly in refining our understanding and enhancing the model's applicability, remains a compelling direction for future research endeavors. }


\vspace{1mm}
\noindent\revise{\textbf{The possibility of using other optimization objectives in {\ML}.} Given that machine learning methods can model a wide range of learning objectives, we have explored two additional representative objectives: tail latency and I/O operations. These have been employed both as learning objectives and evaluation metrics within {\ML}. The results demonstrate their feasibility as optimization objectives. The tail latency we focus on is the 90th percentile, a common metric in SQL databases~\cite{sun13end,wang2020we}. As illustrated in Figure~\ref{fig:obj}, with sampling hours increased to 1.5, we observe approximately 15\% lower latency compared to a well-tuned RocksDB configuration. When optimizing for I/O operations, the improvement is less pronounced, achieving only 8\% lower latency than RocksDB at 1.5 sampling hours. This reduced effectiveness can be attributed to the complex correlations between I/O operations and configurations, primarily due to the randomness introduced by compaction and block cache activities, which can adversely affect the learning effectiveness of {\ML}. Despite these challenges, there remains potential for employing other objectives within the {\ML} framework in certain scenarios, such as those where outlier response times are critical or I/O operations are a bottleneck.}



\section{Related Work}
\label{sec:related}
\noindent\textbf{Machine Learning for Database Systems.} There have been many extensive studies focused on using machine learning to optimize storage systems, particularly for SQL databases~\cite{wang2022wetune,butrovich2022tastes,ma2021mb2}. Ottertune~\cite{van2017automatic} uses machine learning techniques to analyze and tune database configuration knobs and give recommendation settings based on statistics collected from the database. Pavlo \textit{et al.}~\cite{pavlo2017self,pavlo2021make} and Aken \textit{et al.}~\cite{van2021inquiry} introduced the concept of a {\it self-driving} database that can make automated decisions with machine learning models. \revise{Abu-Libdeh \textit{et al.}~\cite{abu2020learned} present several design decisions involved in integrating learned indexes into Bigtable~\cite{chang2008bigtable} and demonstrate significant performance improvements as a result.} Meanwhile, Qtune~\cite{li2019qtune}, CDBTune~\cite{zhang2019end}, and CDBTune+~\cite{zhang2021cdbtune+} are reinforcement learning-based tuning methods for databases. 
There are also studies concentrating on cost estimation~\cite{sun13end,marcus2019plan} and cardinality estimation~\cite{hilprecht2019deepdb,dutt2019selectivity} based on the queries to optimize the plan.

The utilization of machine learning in LSM-tree has rarely been considered, particularly in the context of instance-optimization. Bourbon~\cite{dai2020wisckey} incorporates machine learning into the LSM-tree, which uses learned index~\cite{kraska2018case} to improve the fence pointers. Leaper~\cite{yang2020leaper} employs a learned prefetcher to improve the block cache. \revise{The most recent work, RusKey~\cite{mo2023learning}, is the first to use reinforcement learning to optimize the compaction policy of an LSM-tree.} To the best of our knowledge, we have the first attempt at applying active learning to the instance-optimized LSM-based key-value stores.

\vspace{1mm}
\noindent\textbf{LSM-tree Optimization.} There are numerous studies on LSM-tree optimizations~\cite{alkowaileet2019lsm,ahmad2015compaction,absalyamov2018lightweight,balmau2019silk,balmau2017triad,bortnikov2018accordion,chan2018hashkv,chatterjee2021cosine,dayan2017monkey,dayan2018dostoevsky,dayan2019log,dayan2022spooky,dayan2021chucky,huynh2022endure,golan2015scaling,huang2019x,idreos2019design,kim2020robust,knorr2022proteus,luo2019performance,raju2017pebblesdb,lu2017wisckey,luo2020rosetta,luo2020breaking,ren2017slimdb,sarkar2020lethe,sarkar2022constructing,sears2012blsm,shetty2013building,thonangi2017log,vinccon2018noftl,wang2014efficient,wu2015lsm,zhang2018surf,zhang2020fpga,zhang2018elasticbf,zhu2021reducing,sarkar2023lsm,mun2022lsm,sarkar2022dissecting}. Monkey~\cite{dayan2017monkey} first systematically formulates the theoretical cost models referred in this paper for an LSM-tree by estimating the expected I/O cost. Based on the model, Monkey co-tunes the compaction policy, the memory allocated to the write buffer and the Bloom filter to locate an optimal LSM-tree design with the minimal
I/O cost for a given workload. Dostoevsky~\cite{dayan2018dostoevsky} shows that existing compaction policies (tiering and leveling) cannot fully
trade between read costs and write costs. Therefore, it proposes
Fluid LSM-tree to enable hybrid compaction policies. Cosine~\cite{chatterjee2021cosine}
presents a more meticulous I/O cost model that is aware of workload distribution for key-value stores on the cloud. Endure~\cite{huynh2022endure} introduces a pipeline to jointly tune the compaction policy, size ratio, and allocated memory, and optimize the performance when workload uncertainty is involved. Spooky~\cite{dayan2022spooky} designs a new compaction granulation for LSM-tree by partitioning data more reasonably and cheapening garbage collection. Compared with these works, our study is the first one that employs both theoretical cost analysis and ML model to jointly tune the common parameters in an LSM-tree.

\section{Conclusion}
\label{sec:conclusion}
{We present {\ML}, a complexity-analysis driven ML-aided tuner for LSM-tree based key-value stores. The core idea is to use machine learning to model the actual cost of workloads while using both theoretical analysis and active learning paradigm to help prune the sampling space. Compared to well-tuned RocksDB, {\ML} achieves an average reduction in system latency by 28\%. Moreover, it significantly reduces training costs by up to 90\% compared to conventional ML approaches, thereby enhancing practical usability.}

\begin{acks}
This research is supported by NTU-NAP startup grant (022029-00001) and Singapore MOE AcRF Tier-2 grant MOE-T2EP20223-0004. We thank the anonymous
reviews for their valuable suggestions.

\end{acks}
\bibliographystyle{ACM-Reference-Format}

\begin{thebibliography}{91}


\ifx \showCODEN    \undefined \def \showCODEN     #1{\unskip}     \fi
\ifx \showDOI      \undefined \def \showDOI       #1{#1}\fi
\ifx \showISBNx    \undefined \def \showISBNx     #1{\unskip}     \fi
\ifx \showISBNxiii \undefined \def \showISBNxiii  #1{\unskip}     \fi
\ifx \showISSN     \undefined \def \showISSN      #1{\unskip}     \fi
\ifx \showLCCN     \undefined \def \showLCCN      #1{\unskip}     \fi
\ifx \shownote     \undefined \def \shownote      #1{#1}          \fi
\ifx \showarticletitle \undefined \def \showarticletitle #1{#1}   \fi
\ifx \showURL      \undefined \def \showURL       {\relax}        \fi
\providecommand\bibfield[2]{#2}
\providecommand\bibinfo[2]{#2}
\providecommand\natexlab[1]{#1}
\providecommand\showeprint[2][]{arXiv:#2}

\bibitem[hba(2016)]%
        {hbase}
 \bibinfo{year}{2016}\natexlab{}.
\newblock \bibinfo{title}{{HBase}}.
\newblock \bibinfo{howpublished}{\url{http://hbase.apache.org/}}.
\newblock
\newblock
\shownote{[Online; accessed 12-January-2022]}.


\bibitem[Wir(2021)]%
        {WiredTiger}
 \bibinfo{year}{2021}\natexlab{}.
\newblock \bibinfo{title}{{WiredTiger}}.
\newblock \bibinfo{howpublished}{\url{https://github.com/wiredtiger/wiredtiger}}.
\newblock


\bibitem[Absalyamov et~al\mbox{.}(2018)]%
        {absalyamov2018lightweight}
\bibfield{author}{\bibinfo{person}{Ildar Absalyamov}, \bibinfo{person}{Michael~J Carey}, {and} \bibinfo{person}{Vassilis~J Tsotras}.} \bibinfo{year}{2018}\natexlab{}.
\newblock \showarticletitle{{Lightweight cardinality estimation in LSM-based systems}}. In \bibinfo{booktitle}{\emph{Proceedings of the 2018 International Conference on Management of Data}}. \bibinfo{pages}{841--855}.
\newblock


\bibitem[Abu-Libdeh et~al\mbox{.}(2020)]%
        {abu2020learned}
\bibfield{author}{\bibinfo{person}{Hussam Abu-Libdeh}, \bibinfo{person}{Deniz Alt{\i}nb{\"u}ken}, \bibinfo{person}{Alex Beutel}, \bibinfo{person}{Ed~H Chi}, \bibinfo{person}{Lyric Doshi}, \bibinfo{person}{Tim Kraska}, \bibinfo{person}{Andy Ly}, \bibinfo{person}{Christopher Olston}, {et~al\mbox{.}}} \bibinfo{year}{2020}\natexlab{}.
\newblock \showarticletitle{Learned indexes for a google-scale disk-based database}.
\newblock \bibinfo{journal}{\emph{arXiv preprint arXiv:2012.12501}} (\bibinfo{year}{2020}).
\newblock


\bibitem[Ahmad and Kemme(2015)]%
        {ahmad2015compaction}
\bibfield{author}{\bibinfo{person}{Muhammad~Yousuf Ahmad} {and} \bibinfo{person}{Bettina Kemme}.} \bibinfo{year}{2015}\natexlab{}.
\newblock \showarticletitle{{Compaction management in distributed key-value datastores}}.
\newblock \bibinfo{journal}{\emph{Proceedings of the VLDB Endowment}} \bibinfo{volume}{8}, \bibinfo{number}{8} (\bibinfo{year}{2015}), \bibinfo{pages}{850--861}.
\newblock


\bibitem[Alkowaileet et~al\mbox{.}(2019)]%
        {alkowaileet2019lsm}
\bibfield{author}{\bibinfo{person}{Wail~Y Alkowaileet}, \bibinfo{person}{Sattam Alsubaiee}, {and} \bibinfo{person}{Michael~J Carey}.} \bibinfo{year}{2019}\natexlab{}.
\newblock \showarticletitle{{An LSM-based Tuple Compaction Framework for Apache AsterixDB (Extended Version)}}.
\newblock \bibinfo{journal}{\emph{arXiv preprint arXiv:1910.08185}} (\bibinfo{year}{2019}).
\newblock


\bibitem[Apache(2016)]%
        {cassandra}
\bibfield{author}{\bibinfo{person}{Apache}.} \bibinfo{year}{2016}\natexlab{}.
\newblock \bibinfo{title}{Cassandra}.
\newblock \bibinfo{howpublished}{\url{http://cassandra.apache.org}}.
\newblock


\bibitem[Armstrong et~al\mbox{.}(2013)]%
        {armstrong2013linkbench}
\bibfield{author}{\bibinfo{person}{Timothy~G Armstrong}, \bibinfo{person}{Vamsi Ponnekanti}, \bibinfo{person}{Dhruba Borthakur}, {and} \bibinfo{person}{Mark Callaghan}.} \bibinfo{year}{2013}\natexlab{}.
\newblock \showarticletitle{Linkbench: a database benchmark based on the facebook social graph}. In \bibinfo{booktitle}{\emph{Proceedings of the 2013 ACM SIGMOD International Conference on Management of Data}}. \bibinfo{pages}{1185--1196}.
\newblock


\bibitem[Balmau et~al\mbox{.}(2017)]%
        {balmau2017triad}
\bibfield{author}{\bibinfo{person}{Oana Balmau}, \bibinfo{person}{Diego Didona}, \bibinfo{person}{Rachid Guerraoui}, \bibinfo{person}{Willy Zwaenepoel}, \bibinfo{person}{Huapeng Yuan}, \bibinfo{person}{Aashray Arora}, \bibinfo{person}{Karan Gupta}, {and} \bibinfo{person}{Pavan Konka}.} \bibinfo{year}{2017}\natexlab{}.
\newblock \showarticletitle{$\{$TRIAD$\}$: Creating Synergies Between Memory, Disk and Log in Log Structured Key-Value Stores}. In \bibinfo{booktitle}{\emph{2017 $\{$USENIX$\}$ Annual Technical Conference ($\{$USENIX$\}$$\{$ATC$\}$ 17)}}. \bibinfo{pages}{363--375}.
\newblock


\bibitem[Balmau et~al\mbox{.}(2019)]%
        {balmau2019silk}
\bibfield{author}{\bibinfo{person}{Oana Balmau}, \bibinfo{person}{Florin Dinu}, \bibinfo{person}{Willy Zwaenepoel}, \bibinfo{person}{Karan Gupta}, \bibinfo{person}{Ravishankar Chandhiramoorthi}, {and} \bibinfo{person}{Diego Didona}.} \bibinfo{year}{2019}\natexlab{}.
\newblock \showarticletitle{{SILK: Preventing Latency Spikes in Log-Structured Merge Key-Value Stores.}}. In \bibinfo{booktitle}{\emph{USENIX Annual Technical Conference}}. \bibinfo{pages}{753--766}.
\newblock


\bibitem[Berthet(2017)]%
        {berthet2017approximation}
\bibfield{author}{\bibinfo{person}{Christian Berthet}.} \bibinfo{year}{2017}\natexlab{}.
\newblock \showarticletitle{{Approximation of LRU caches miss rate: Application to power-law popularities}}.
\newblock \bibinfo{journal}{\emph{arXiv preprint arXiv:1705.10738}} (\bibinfo{year}{2017}).
\newblock


\bibitem[Bortnikov et~al\mbox{.}(2018)]%
        {bortnikov2018accordion}
\bibfield{author}{\bibinfo{person}{Edward Bortnikov}, \bibinfo{person}{Anastasia Braginsky}, \bibinfo{person}{Eshcar Hillel}, \bibinfo{person}{Idit Keidar}, {and} \bibinfo{person}{Gali Sheffi}.} \bibinfo{year}{2018}\natexlab{}.
\newblock \showarticletitle{{Accordion: Better memory organization for LSM key-value stores}}.
\newblock \bibinfo{journal}{\emph{Proceedings of the VLDB Endowment}} \bibinfo{volume}{11}, \bibinfo{number}{12} (\bibinfo{year}{2018}), \bibinfo{pages}{1863--1875}.
\newblock


\bibitem[Butrovich et~al\mbox{.}(2022)]%
        {butrovich2022tastes}
\bibfield{author}{\bibinfo{person}{Matthew Butrovich}, \bibinfo{person}{Wan~Shen Lim}, \bibinfo{person}{Lin Ma}, \bibinfo{person}{John Rollinson}, \bibinfo{person}{William Zhang}, \bibinfo{person}{Yu Xia}, {and} \bibinfo{person}{Andrew Pavlo}.} \bibinfo{year}{2022}\natexlab{}.
\newblock \showarticletitle{Tastes Great! Less Filling! High Performance and Accurate Training Data Collection for Self-Driving Database Management Systems}. In \bibinfo{booktitle}{\emph{Proceedings of the 2022 International Conference on Management of Data}}. \bibinfo{pages}{617--630}.
\newblock


\bibitem[Cao et~al\mbox{.}(2013)]%
        {cao2013logkv}
\bibfield{author}{\bibinfo{person}{Zhao Cao}, \bibinfo{person}{Shimin Chen}, \bibinfo{person}{Feifei Li}, \bibinfo{person}{Min Wang}, {and} \bibinfo{person}{X~Sean Wang}.} \bibinfo{year}{2013}\natexlab{}.
\newblock \showarticletitle{{LogKV: Exploiting key-value stores for event log processing}}. In \bibinfo{booktitle}{\emph{Proc. Conf. Innovative Data Syst. Res}}.
\newblock


\bibitem[Chan et~al\mbox{.}(2018)]%
        {chan2018hashkv}
\bibfield{author}{\bibinfo{person}{Helen~HW Chan}, \bibinfo{person}{Yongkun Li}, \bibinfo{person}{Patrick~PC Lee}, {and} \bibinfo{person}{Yinlong Xu}.} \bibinfo{year}{2018}\natexlab{}.
\newblock \showarticletitle{{Hashkv: Enabling efficient updates in $\{$KV$\}$ storage via hashing}}. In \bibinfo{booktitle}{\emph{2018 $\{$USENIX$\}$ Annual Technical Conference ($\{$USENIX$\}$$\{$ATC$\}$ 18)}}. \bibinfo{pages}{1007--1019}.
\newblock


\bibitem[Chang et~al\mbox{.}(2008)]%
        {chang2008bigtable}
\bibfield{author}{\bibinfo{person}{Fay Chang}, \bibinfo{person}{Jeffrey Dean}, \bibinfo{person}{Sanjay Ghemawat}, \bibinfo{person}{Wilson~C Hsieh}, \bibinfo{person}{Deborah~A Wallach}, \bibinfo{person}{Mike Burrows}, \bibinfo{person}{Tushar Chandra}, \bibinfo{person}{Andrew Fikes}, {and} \bibinfo{person}{Robert~E Gruber}.} \bibinfo{year}{2008}\natexlab{}.
\newblock \showarticletitle{Bigtable: A distributed storage system for structured data}.
\newblock \bibinfo{journal}{\emph{ACM Transactions on Computer Systems (TOCS)}} \bibinfo{volume}{26}, \bibinfo{number}{2} (\bibinfo{year}{2008}), \bibinfo{pages}{1--26}.
\newblock


\bibitem[Chatterjee et~al\mbox{.}(2021)]%
        {chatterjee2021cosine}
\bibfield{author}{\bibinfo{person}{Subarna Chatterjee}, \bibinfo{person}{Meena Jagadeesan}, \bibinfo{person}{Wilson Qin}, {and} \bibinfo{person}{Stratos Idreos}.} \bibinfo{year}{2021}\natexlab{}.
\newblock \showarticletitle{Cosine: a cloud-cost optimized self-designing key-value storage engine}.
\newblock \bibinfo{journal}{\emph{Proceedings of the VLDB Endowment}} \bibinfo{volume}{15}, \bibinfo{number}{1} (\bibinfo{year}{2021}), \bibinfo{pages}{112--126}.
\newblock


\bibitem[Chen et~al\mbox{.}(2016)]%
        {chen2016realtime}
\bibfield{author}{\bibinfo{person}{Guoqiang~Jerry Chen}, \bibinfo{person}{Janet~L Wiener}, \bibinfo{person}{Shridhar Iyer}, \bibinfo{person}{Anshul Jaiswal}, \bibinfo{person}{Ran Lei}, \bibinfo{person}{Nikhil Simha}, \bibinfo{person}{Wei Wang}, \bibinfo{person}{Kevin Wilfong}, \bibinfo{person}{Tim Williamson}, {and} \bibinfo{person}{Serhat Yilmaz}.} \bibinfo{year}{2016}\natexlab{}.
\newblock \showarticletitle{Realtime data processing at facebook}. In \bibinfo{booktitle}{\emph{Proceedings of the 2016 International Conference on Management of Data}}. \bibinfo{pages}{1087--1098}.
\newblock


\bibitem[Chen et~al\mbox{.}(2015)]%
        {chen2015xgboost}
\bibfield{author}{\bibinfo{person}{Tianqi Chen}, \bibinfo{person}{Tong He}, \bibinfo{person}{Michael Benesty}, \bibinfo{person}{Vadim Khotilovich}, \bibinfo{person}{Yuan Tang}, \bibinfo{person}{Hyunsu Cho}, \bibinfo{person}{Kailong Chen}, \bibinfo{person}{Rory Mitchell}, \bibinfo{person}{Ignacio Cano}, \bibinfo{person}{Tianyi Zhou}, {et~al\mbox{.}}} \bibinfo{year}{2015}\natexlab{}.
\newblock \showarticletitle{Xgboost: extreme gradient boosting}.
\newblock \bibinfo{journal}{\emph{R package version 0.4-2}} \bibinfo{volume}{1}, \bibinfo{number}{4} (\bibinfo{year}{2015}), \bibinfo{pages}{1--4}.
\newblock


\bibitem[Cooper et~al\mbox{.}(2010)]%
        {cooper2010benchmarking}
\bibfield{author}{\bibinfo{person}{Brian~F Cooper}, \bibinfo{person}{Adam Silberstein}, \bibinfo{person}{Erwin Tam}, \bibinfo{person}{Raghu Ramakrishnan}, {and} \bibinfo{person}{Russell Sears}.} \bibinfo{year}{2010}\natexlab{}.
\newblock \showarticletitle{{Benchmarking cloud serving systems with YCSB}}. In \bibinfo{booktitle}{\emph{Proceedings of the 1st ACM symposium on Cloud computing}}. \bibinfo{pages}{143--154}.
\newblock


\bibitem[Dai et~al\mbox{.}(2020)]%
        {dai2020wisckey}
\bibfield{author}{\bibinfo{person}{Yifan Dai}, \bibinfo{person}{Yien Xu}, \bibinfo{person}{Aishwarya Ganesan}, \bibinfo{person}{Ramnatthan Alagappan}, \bibinfo{person}{Brian Kroth}, \bibinfo{person}{Andrea~C Arpaci-Dusseau}, {and} \bibinfo{person}{Remzi~H Arpaci-Dusseau}.} \bibinfo{year}{2020}\natexlab{}.
\newblock \showarticletitle{From wisckey to bourbon: A learned index for log-structured merge trees}. In \bibinfo{booktitle}{\emph{Proceedings of the 14th USENIX Conference on Operating Systems Design and Implementation}}. \bibinfo{pages}{155--171}.
\newblock


\bibitem[Dayan et~al\mbox{.}(2017)]%
        {dayan2017monkey}
\bibfield{author}{\bibinfo{person}{Niv Dayan}, \bibinfo{person}{Manos Athanassoulis}, {and} \bibinfo{person}{Stratos Idreos}.} \bibinfo{year}{2017}\natexlab{}.
\newblock \showarticletitle{{Monkey: Optimal navigable key-value store}}. In \bibinfo{booktitle}{\emph{Proceedings of the 2017 ACM International Conference on Management of Data}}. \bibinfo{pages}{79--94}.
\newblock


\bibitem[Dayan and Idreos(2018)]%
        {dayan2018dostoevsky}
\bibfield{author}{\bibinfo{person}{Niv Dayan} {and} \bibinfo{person}{Stratos Idreos}.} \bibinfo{year}{2018}\natexlab{}.
\newblock \showarticletitle{{Dostoevsky: Better space-time trade-offs for LSM-tree based key-value stores via adaptive removal of superfluous merging}}. In \bibinfo{booktitle}{\emph{Proceedings of the 2018 International Conference on Management of Data}}. \bibinfo{pages}{505--520}.
\newblock


\bibitem[Dayan and Twitto(2021)]%
        {dayan2021chucky}
\bibfield{author}{\bibinfo{person}{Niv Dayan} {and} \bibinfo{person}{Moshe Twitto}.} \bibinfo{year}{2021}\natexlab{}.
\newblock \showarticletitle{{Chucky: A succinct cuckoo filter for LSM-tree}}. In \bibinfo{booktitle}{\emph{Proceedings of the 2021 International Conference on Management of Data}}. \bibinfo{pages}{365--378}.
\newblock


\bibitem[Dayan et~al\mbox{.}(2022)]%
        {dayan2022spooky}
\bibfield{author}{\bibinfo{person}{Niv Dayan}, \bibinfo{person}{Tamar Weiss}, \bibinfo{person}{Shmuel Dashevsky}, \bibinfo{person}{Michael Pan}, \bibinfo{person}{Edward Bortnikov}, {and} \bibinfo{person}{Moshe Twitto}.} \bibinfo{year}{2022}\natexlab{}.
\newblock \showarticletitle{{Spooky: granulating LSM-tree compactions correctly}}.
\newblock \bibinfo{journal}{\emph{Proceedings of the VLDB Endowment}} \bibinfo{volume}{15}, \bibinfo{number}{11} (\bibinfo{year}{2022}), \bibinfo{pages}{3071--3084}.
\newblock


\bibitem[{Dayan, Niv and Idreos, Stratos}(2019)]%
        {dayan2019log}
\bibfield{author}{\bibinfo{person}{{Dayan, Niv and Idreos, Stratos}}.} \bibinfo{year}{2019}\natexlab{}.
\newblock \showarticletitle{The log-structured merge-bush \& the wacky continuum}. In \bibinfo{booktitle}{\emph{Proceedings of the 2019 International Conference on Management of Data}}. \bibinfo{pages}{449--466}.
\newblock


\bibitem[DeCandia et~al\mbox{.}(2007)]%
        {decandia2007dynamo}
\bibfield{author}{\bibinfo{person}{Giuseppe DeCandia}, \bibinfo{person}{D Hastorun}, \bibinfo{person}{Madan Jampani}, \bibinfo{person}{Gunavardhan Kakulapati}, \bibinfo{person}{Avinash Lakshman}, \bibinfo{person}{S~Sivasubramanian~A Pilchin}, \bibinfo{person}{P Vosshall}, {and} \bibinfo{person}{W Vogels}.} \bibinfo{year}{2007}\natexlab{}.
\newblock \bibinfo{title}{{Dynamo: Amazon’s Highly Available Key-Value Store. SOSP, 2007}}.
\newblock
\newblock


\bibitem[Duan et~al\mbox{.}(2009)]%
        {duan2009tuning}
\bibfield{author}{\bibinfo{person}{Songyun Duan}, \bibinfo{person}{Vamsidhar Thummala}, {and} \bibinfo{person}{Shivnath Babu}.} \bibinfo{year}{2009}\natexlab{}.
\newblock \showarticletitle{Tuning database configuration parameters with ituned}.
\newblock \bibinfo{journal}{\emph{Proceedings of the VLDB Endowment}} \bibinfo{volume}{2}, \bibinfo{number}{1} (\bibinfo{year}{2009}), \bibinfo{pages}{1246--1257}.
\newblock


\bibitem[Dutt et~al\mbox{.}(2019)]%
        {dutt2019selectivity}
\bibfield{author}{\bibinfo{person}{Anshuman Dutt}, \bibinfo{person}{Chi Wang}, \bibinfo{person}{Azade Nazi}, \bibinfo{person}{Srikanth Kandula}, \bibinfo{person}{Vivek Narasayya}, {and} \bibinfo{person}{Surajit Chaudhuri}.} \bibinfo{year}{2019}\natexlab{}.
\newblock \showarticletitle{Selectivity estimation for range predicates using lightweight models}.
\newblock \bibinfo{journal}{\emph{Proceedings of the VLDB Endowment}} \bibinfo{volume}{12}, \bibinfo{number}{9} (\bibinfo{year}{2019}), \bibinfo{pages}{1044--1057}.
\newblock


\bibitem[Facebook(2016)]%
        {rocksdb}
\bibfield{author}{\bibinfo{person}{Facebook}.} \bibinfo{year}{2016}\natexlab{}.
\newblock \bibinfo{title}{{RocksDB}}.
\newblock \bibinfo{howpublished}{\url{https://github.com/facebook/rocksdb}}.
\newblock


\bibitem[Golan-Gueta et~al\mbox{.}(2015)]%
        {golan2015scaling}
\bibfield{author}{\bibinfo{person}{Guy Golan-Gueta}, \bibinfo{person}{Edward Bortnikov}, \bibinfo{person}{Eshcar Hillel}, {and} \bibinfo{person}{Idit Keidar}.} \bibinfo{year}{2015}\natexlab{}.
\newblock \showarticletitle{Scaling concurrent log-structured data stores}. In \bibinfo{booktitle}{\emph{Proceedings of the Tenth European Conference on Computer Systems}}. \bibinfo{pages}{1--14}.
\newblock


\bibitem[Google(2016)]%
        {google-leveldb}
\bibfield{author}{\bibinfo{person}{Google}.} \bibinfo{year}{2016}\natexlab{}.
\newblock \bibinfo{title}{{LevelDB}}.
\newblock \bibinfo{howpublished}{\url{https://github.com/google/leveldb/}}.
\newblock


\bibitem[Gu et~al\mbox{.}(2023)]%
        {gu2023rlr}
\bibfield{author}{\bibinfo{person}{Tu Gu}, \bibinfo{person}{Kaiyu Feng}, \bibinfo{person}{Gao Cong}, \bibinfo{person}{Cheng Long}, \bibinfo{person}{Zheng Wang}, {and} \bibinfo{person}{Sheng Wang}.} \bibinfo{year}{2023}\natexlab{}.
\newblock \showarticletitle{The RLR-Tree: A Reinforcement Learning Based R-Tree for Spatial Data}.
\newblock \bibinfo{journal}{\emph{Proceedings of the ACM on Management of Data}} \bibinfo{volume}{1}, \bibinfo{number}{1} (\bibinfo{year}{2023}), \bibinfo{pages}{1--26}.
\newblock


\bibitem[Harris et~al\mbox{.}(2020)]%
        {harris2020array}
\bibfield{author}{\bibinfo{person}{Charles~R Harris}, \bibinfo{person}{K~Jarrod Millman}, \bibinfo{person}{St{\'e}fan~J Van Der~Walt}, \bibinfo{person}{Ralf Gommers}, \bibinfo{person}{Pauli Virtanen}, \bibinfo{person}{David Cournapeau}, \bibinfo{person}{Eric Wieser}, \bibinfo{person}{Julian Taylor}, \bibinfo{person}{Sebastian Berg}, \bibinfo{person}{Nathaniel~J Smith}, {et~al\mbox{.}}} \bibinfo{year}{2020}\natexlab{}.
\newblock \showarticletitle{{Array programming with NumPy}}.
\newblock \bibinfo{journal}{\emph{Nature}} \bibinfo{volume}{585}, \bibinfo{number}{7825} (\bibinfo{year}{2020}), \bibinfo{pages}{357--362}.
\newblock


\bibitem[Hastie et~al\mbox{.}(2009)]%
        {hastie2009elements}
\bibfield{author}{\bibinfo{person}{Trevor Hastie}, \bibinfo{person}{Robert Tibshirani}, \bibinfo{person}{Jerome~H Friedman}, {and} \bibinfo{person}{Jerome~H Friedman}.} \bibinfo{year}{2009}\natexlab{}.
\newblock \bibinfo{booktitle}{\emph{The elements of statistical learning: data mining, inference, and prediction}}. Vol.~\bibinfo{volume}{2}.
\newblock \bibinfo{publisher}{Springer}.
\newblock


\bibitem[Hilprecht et~al\mbox{.}(2019)]%
        {hilprecht2019deepdb}
\bibfield{author}{\bibinfo{person}{Benjamin Hilprecht}, \bibinfo{person}{Andreas Schmidt}, \bibinfo{person}{Moritz Kulessa}, \bibinfo{person}{Alejandro Molina}, \bibinfo{person}{Kristian Kersting}, {and} \bibinfo{person}{Carsten Binnig}.} \bibinfo{year}{2019}\natexlab{}.
\newblock \showarticletitle{{DeepDB: Learn from Data, not from Queries!}}
\newblock \bibinfo{journal}{\emph{Proceedings of the VLDB Endowment}} \bibinfo{volume}{13}, \bibinfo{number}{7} (\bibinfo{year}{2019}).
\newblock


\bibitem[Hu et~al\mbox{.}(2021)]%
        {hu2021model}
\bibfield{author}{\bibinfo{person}{Xia Hu}, \bibinfo{person}{Lingyang Chu}, \bibinfo{person}{Jian Pei}, \bibinfo{person}{Weiqing Liu}, {and} \bibinfo{person}{Jiang Bian}.} \bibinfo{year}{2021}\natexlab{}.
\newblock \showarticletitle{{Model complexity of deep learning: A survey}}.
\newblock \bibinfo{journal}{\emph{Knowledge and Information Systems}}  \bibinfo{volume}{63} (\bibinfo{year}{2021}), \bibinfo{pages}{2585--2619}.
\newblock


\bibitem[Huang et~al\mbox{.}(2018)]%
        {huang2018optimization}
\bibfield{author}{\bibinfo{person}{Enhui Huang}, \bibinfo{person}{Liping Peng}, \bibinfo{person}{Luciano Di~Palma}, \bibinfo{person}{Ahmed Abdelkafi}, \bibinfo{person}{Anna Liu}, {and} \bibinfo{person}{Yanlei Diao}.} \bibinfo{year}{2018}\natexlab{}.
\newblock \emph{\bibinfo{title}{Optimization for active learning-based interactive database exploration}}.
\newblock \bibinfo{thesistype}{Ph.\,D. Dissertation}. \bibinfo{school}{Ecole Polytechnique; University of Massachusetts Amherst}.
\newblock


\bibitem[Huang et~al\mbox{.}(2019)]%
        {huang2019x}
\bibfield{author}{\bibinfo{person}{Gui Huang}, \bibinfo{person}{Xuntao Cheng}, \bibinfo{person}{Jianying Wang}, \bibinfo{person}{Yujie Wang}, \bibinfo{person}{Dengcheng He}, \bibinfo{person}{Tieying Zhang}, \bibinfo{person}{Feifei Li}, \bibinfo{person}{Sheng Wang}, \bibinfo{person}{Wei Cao}, {and} \bibinfo{person}{Qiang Li}.} \bibinfo{year}{2019}\natexlab{}.
\newblock \showarticletitle{{X-Engine: An optimized storage engine for large-scale E-commerce transaction processing}}. In \bibinfo{booktitle}{\emph{Proceedings of the 2019 International Conference on Management of Data}}. \bibinfo{pages}{651--665}.
\newblock


\bibitem[Huynh et~al\mbox{.}(2022)]%
        {huynh2022endure}
\bibfield{author}{\bibinfo{person}{Andy Huynh}, \bibinfo{person}{Harshal~A Chaudhari}, \bibinfo{person}{Evimaria Terzi}, {and} \bibinfo{person}{Manos Athanassoulis}.} \bibinfo{year}{2022}\natexlab{}.
\newblock \showarticletitle{Endure: a robust tuning paradigm for LSM trees under workload uncertainty}.
\newblock \bibinfo{journal}{\emph{Proceedings of the VLDB Endowment}} \bibinfo{volume}{15}, \bibinfo{number}{8} (\bibinfo{year}{2022}), \bibinfo{pages}{1605--1618}.
\newblock


\bibitem[Huynh et~al\mbox{.}(2023)]%
        {huynh2023flexibility}
\bibfield{author}{\bibinfo{person}{Andy Huynh}, \bibinfo{person}{Harshal~A. Chaudhari}, \bibinfo{person}{Evimaria Terzi}, {and} \bibinfo{person}{Manos Athanassoulis}.} \bibinfo{year}{2023}\natexlab{}.
\newblock \bibinfo{title}{Towards Flexibility and Robustness of LSM Trees}.
\newblock
\newblock
\showeprint[arxiv]{2311.10005}~[cs.DB]


\bibitem[Idreos et~al\mbox{.}(2019)]%
        {idreos2019design}
\bibfield{author}{\bibinfo{person}{Stratos Idreos}, \bibinfo{person}{Niv Dayan}, \bibinfo{person}{Wilson Qin}, \bibinfo{person}{Mali Akmanalp}, \bibinfo{person}{Sophie Hilgard}, \bibinfo{person}{Andrew Ross}, \bibinfo{person}{James Lennon}, \bibinfo{person}{Varun Jain}, \bibinfo{person}{Harshita Gupta}, \bibinfo{person}{David Li}, {et~al\mbox{.}}} \bibinfo{year}{2019}\natexlab{}.
\newblock \showarticletitle{Design Continuums and the Path Toward Self-Designing Key-Value Stores that Know and Learn.}. In \bibinfo{booktitle}{\emph{CIDR}}.
\newblock


\bibitem[Jannen et~al\mbox{.}(2015)]%
        {jannen2015betrfs}
\bibfield{author}{\bibinfo{person}{William Jannen}, \bibinfo{person}{Jun Yuan}, \bibinfo{person}{Yang Zhan}, \bibinfo{person}{Amogh Akshintala}, \bibinfo{person}{John Esmet}, \bibinfo{person}{Yizheng Jiao}, \bibinfo{person}{Ankur Mittal}, \bibinfo{person}{Prashant Pandey}, \bibinfo{person}{Phaneendra Reddy}, \bibinfo{person}{Leif Walsh}, {et~al\mbox{.}}} \bibinfo{year}{2015}\natexlab{}.
\newblock \showarticletitle{{BetrFS: A Right-Optimized Write-Optimized File System.}}. In \bibinfo{booktitle}{\emph{FAST}}, Vol.~\bibinfo{volume}{15}. \bibinfo{pages}{301--315}.
\newblock


\bibitem[Juba and Le(2019)]%
        {juba2019precision}
\bibfield{author}{\bibinfo{person}{Brendan Juba} {and} \bibinfo{person}{Hai~S Le}.} \bibinfo{year}{2019}\natexlab{}.
\newblock \showarticletitle{{Precision-recall versus accuracy and the role of large data sets}}. In \bibinfo{booktitle}{\emph{Proceedings of the AAAI conference on artificial intelligence}}, Vol.~\bibinfo{volume}{33}. \bibinfo{pages}{4039--4048}.
\newblock


\bibitem[Kim et~al\mbox{.}(2020)]%
        {kim2020robust}
\bibfield{author}{\bibinfo{person}{Taewoo Kim}, \bibinfo{person}{Alexander Behm}, \bibinfo{person}{Michael Blow}, \bibinfo{person}{Vinayak Borkar}, \bibinfo{person}{Yingyi Bu}, \bibinfo{person}{Michael~J Carey}, \bibinfo{person}{Murtadha Hubail}, \bibinfo{person}{Shiva Jahangiri}, \bibinfo{person}{Jianfeng Jia}, \bibinfo{person}{Chen Li}, {et~al\mbox{.}}} \bibinfo{year}{2020}\natexlab{}.
\newblock \showarticletitle{{Robust and efficient memory management in Apache AsterixDB}}.
\newblock \bibinfo{journal}{\emph{Software: Practice and Experience}} \bibinfo{volume}{50}, \bibinfo{number}{7} (\bibinfo{year}{2020}), \bibinfo{pages}{1114--1151}.
\newblock


\bibitem[Knorr et~al\mbox{.}(2022)]%
        {knorr2022proteus}
\bibfield{author}{\bibinfo{person}{Eric~R Knorr}, \bibinfo{person}{Baptiste Lemaire}, \bibinfo{person}{Andrew Lim}, \bibinfo{person}{Siqiang Luo}, \bibinfo{person}{Huanchen Zhang}, \bibinfo{person}{Stratos Idreos}, {and} \bibinfo{person}{Michael Mitzenmacher}.} \bibinfo{year}{2022}\natexlab{}.
\newblock \showarticletitle{{Proteus: A Self-Designing Range Filter}}. In \bibinfo{booktitle}{\emph{Proceedings of the 2022 International Conference on Management of Data}}. \bibinfo{pages}{1670--1684}.
\newblock


\bibitem[Kraska et~al\mbox{.}(2018)]%
        {kraska2018case}
\bibfield{author}{\bibinfo{person}{Tim Kraska}, \bibinfo{person}{Alex Beutel}, \bibinfo{person}{Ed~H Chi}, \bibinfo{person}{Jeffrey Dean}, {and} \bibinfo{person}{Neoklis Polyzotis}.} \bibinfo{year}{2018}\natexlab{}.
\newblock \showarticletitle{The case for learned index structures}. In \bibinfo{booktitle}{\emph{Proceedings of the 2018 international conference on management of data}}. \bibinfo{pages}{489--504}.
\newblock


\bibitem[Li et~al\mbox{.}(2019)]%
        {li2019qtune}
\bibfield{author}{\bibinfo{person}{Guoliang Li}, \bibinfo{person}{Xuanhe Zhou}, \bibinfo{person}{Shifu Li}, {and} \bibinfo{person}{Bo Gao}.} \bibinfo{year}{2019}\natexlab{}.
\newblock \showarticletitle{{Qtune: A query-aware database tuning system with deep reinforcement learning}}.
\newblock \bibinfo{journal}{\emph{Proceedings of the VLDB Endowment}} \bibinfo{volume}{12}, \bibinfo{number}{12} (\bibinfo{year}{2019}), \bibinfo{pages}{2118--2130}.
\newblock


\bibitem[Lin et~al\mbox{.}(2022)]%
        {lin2022learning}
\bibfield{author}{\bibinfo{person}{Honghao Lin}, \bibinfo{person}{Tian Luo}, {and} \bibinfo{person}{David Woodruff}.} \bibinfo{year}{2022}\natexlab{}.
\newblock \showarticletitle{Learning augmented binary search trees}. In \bibinfo{booktitle}{\emph{International Conference on Machine Learning}}. PMLR, \bibinfo{pages}{13431--13440}.
\newblock


\bibitem[Liu et~al\mbox{.}(2024)]%
        {liu2024structural}
\bibfield{author}{\bibinfo{person}{Junfeng Liu}, \bibinfo{person}{Fan Wang}, \bibinfo{person}{Dingheng Mo}, {and} \bibinfo{person}{Siqiang Luo}.} \bibinfo{year}{2024}\natexlab{}.
\newblock \showarticletitle{Structural Designs Meet Optimality: Exploring Optimized LSM-tree Structures in A Colossal Configuration Space}. In \bibinfo{booktitle}{\emph{Proceedings of the 2024 ACM SIGMOD International Conference on Management of Data (SIGMOD)}}. ACM.
\newblock


\bibitem[Lu et~al\mbox{.}(2017)]%
        {lu2017wisckey}
\bibfield{author}{\bibinfo{person}{Lanyue Lu}, \bibinfo{person}{Thanumalayan~Sankaranarayana Pillai}, \bibinfo{person}{Hariharan Gopalakrishnan}, \bibinfo{person}{Andrea~C Arpaci-Dusseau}, {and} \bibinfo{person}{Remzi~H Arpaci-Dusseau}.} \bibinfo{year}{2017}\natexlab{}.
\newblock \showarticletitle{{Wisckey: Separating keys from values in ssd-conscious storage}}.
\newblock \bibinfo{journal}{\emph{ACM Transactions on Storage (TOS)}} \bibinfo{volume}{13}, \bibinfo{number}{1} (\bibinfo{year}{2017}), \bibinfo{pages}{1--28}.
\newblock


\bibitem[Luo and Carey(2019)]%
        {luo2019performance}
\bibfield{author}{\bibinfo{person}{Chen Luo} {and} \bibinfo{person}{Michael~J Carey}.} \bibinfo{year}{2019}\natexlab{}.
\newblock \showarticletitle{{On performance stability in LSM-based storage systems (extended version)}}.
\newblock \bibinfo{journal}{\emph{arXiv preprint arXiv:1906.09667}} (\bibinfo{year}{2019}).
\newblock


\bibitem[Luo and Carey(2020)]%
        {luo2020breaking}
\bibfield{author}{\bibinfo{person}{Chen Luo} {and} \bibinfo{person}{Michael~J Carey}.} \bibinfo{year}{2020}\natexlab{}.
\newblock \showarticletitle{{Breaking down memory walls: Adaptive memory management in LSM-based storage systems}}.
\newblock \bibinfo{journal}{\emph{Proceedings of the VLDB Endowment}} \bibinfo{volume}{14}, \bibinfo{number}{3} (\bibinfo{year}{2020}), \bibinfo{pages}{241--254}.
\newblock


\bibitem[Luo et~al\mbox{.}(2020)]%
        {luo2020rosetta}
\bibfield{author}{\bibinfo{person}{Siqiang Luo}, \bibinfo{person}{Subarna Chatterjee}, \bibinfo{person}{Rafael Ketsetsidis}, \bibinfo{person}{Niv Dayan}, \bibinfo{person}{Wilson Qin}, {and} \bibinfo{person}{Stratos Idreos}.} \bibinfo{year}{2020}\natexlab{}.
\newblock \showarticletitle{{Rosetta: A robust space-time optimized range filter for key-value stores}}. In \bibinfo{booktitle}{\emph{Proceedings of the 2020 ACM SIGMOD International Conference on Management of Data}}. \bibinfo{pages}{2071--2086}.
\newblock


\bibitem[Ma et~al\mbox{.}(2020)]%
        {ma2020active}
\bibfield{author}{\bibinfo{person}{Lin Ma}, \bibinfo{person}{Bailu Ding}, \bibinfo{person}{Sudipto Das}, {and} \bibinfo{person}{Adith Swaminathan}.} \bibinfo{year}{2020}\natexlab{}.
\newblock \showarticletitle{Active learning for ML enhanced database systems}. In \bibinfo{booktitle}{\emph{Proceedings of the 2020 ACM SIGMOD International Conference on Management of Data}}. \bibinfo{pages}{175--191}.
\newblock


\bibitem[Ma et~al\mbox{.}(2021)]%
        {ma2021mb2}
\bibfield{author}{\bibinfo{person}{Lin Ma}, \bibinfo{person}{William Zhang}, \bibinfo{person}{Jie Jiao}, \bibinfo{person}{Wuwen Wang}, \bibinfo{person}{Matthew Butrovich}, \bibinfo{person}{Wan~Shen Lim}, \bibinfo{person}{Prashanth Menon}, {and} \bibinfo{person}{Andrew Pavlo}.} \bibinfo{year}{2021}\natexlab{}.
\newblock \showarticletitle{{MB2: decomposed behavior modeling for self-driving database management systems}}. In \bibinfo{booktitle}{\emph{Proceedings of the 2021 International Conference on Management of Data}}. \bibinfo{pages}{1248--1261}.
\newblock


\bibitem[Marcus and Papaemmanouil(2019)]%
        {marcus2019plan}
\bibfield{author}{\bibinfo{person}{Ryan Marcus} {and} \bibinfo{person}{Olga Papaemmanouil}.} \bibinfo{year}{2019}\natexlab{}.
\newblock \showarticletitle{Plan-structured deep neural network models for query performance prediction}.
\newblock \bibinfo{journal}{\emph{Proceedings of the VLDB Endowment}} \bibinfo{volume}{12}, \bibinfo{number}{11} (\bibinfo{year}{2019}), \bibinfo{pages}{1733--1746}.
\newblock


\bibitem[Mo et~al\mbox{.}(2023)]%
        {mo2023learning}
\bibfield{author}{\bibinfo{person}{Dingheng Mo}, \bibinfo{person}{Fanchao Chen}, \bibinfo{person}{Siqiang Luo}, {and} \bibinfo{person}{Caihua Shan}.} \bibinfo{year}{2023}\natexlab{}.
\newblock \showarticletitle{Learning to Optimize LSM-trees: Towards A Reinforcement Learning based Key-Value Store for Dynamic Workloads}.
\newblock \bibinfo{journal}{\emph{Proceedings of the ACM on Management of Data}} \bibinfo{volume}{1}, \bibinfo{number}{3} (\bibinfo{year}{2023}), \bibinfo{pages}{1--25}.
\newblock


\bibitem[Mozafari et~al\mbox{.}(2012)]%
        {mozafari2012active}
\bibfield{author}{\bibinfo{person}{Barzan Mozafari}, \bibinfo{person}{Purnamrita Sarkar}, \bibinfo{person}{Michael~J Franklin}, \bibinfo{person}{Michael~I Jordan}, {and} \bibinfo{person}{Samuel Madden}.} \bibinfo{year}{2012}\natexlab{}.
\newblock \showarticletitle{Active learning for crowd-sourced databases}.
\newblock \bibinfo{journal}{\emph{arXiv preprint arXiv:1209.3686}} (\bibinfo{year}{2012}).
\newblock


\bibitem[Mun et~al\mbox{.}(2022)]%
        {mun2022lsm}
\bibfield{author}{\bibinfo{person}{Ju~Hyoung Mun}, \bibinfo{person}{Zichen Zhu}, \bibinfo{person}{Aneesh Raman}, {and} \bibinfo{person}{Manos Athanassoulis}.} \bibinfo{year}{2022}\natexlab{}.
\newblock \showarticletitle{{LSM-Trees Under (Memory) Pressure}}. In \bibinfo{booktitle}{\emph{Proceedings of the International Workshop on Accelerating Data Management Systems Using Modern Processor and Storage Architectures (ADMS)}}.
\newblock


\bibitem[Nogueira(14  )]%
        {bo2014}
\bibfield{author}{\bibinfo{person}{Fernando Nogueira}.} \bibinfo{year}{2014--}\natexlab{}.
\newblock \bibinfo{title}{{Bayesian Optimization}: Open source constrained global optimization tool for {Python}}.
\newblock
\newblock
\urldef\tempurl%
\url{https://github.com/fmfn/BayesianOptimization}
\showURL{%
\tempurl}


\bibitem[O’Neil et~al\mbox{.}(1996)]%
        {o1996log}
\bibfield{author}{\bibinfo{person}{Patrick O’Neil}, \bibinfo{person}{Edward Cheng}, \bibinfo{person}{Dieter Gawlick}, {and} \bibinfo{person}{Elizabeth O’Neil}.} \bibinfo{year}{1996}\natexlab{}.
\newblock \showarticletitle{{The log-structured merge-tree (LSM-tree)}}.
\newblock \bibinfo{journal}{\emph{Acta Informatica}}  \bibinfo{volume}{33} (\bibinfo{year}{1996}), \bibinfo{pages}{351--385}.
\newblock


\bibitem[Paszke et~al\mbox{.}(2019)]%
        {paszke2019pytorch}
\bibfield{author}{\bibinfo{person}{Adam Paszke}, \bibinfo{person}{Sam Gross}, \bibinfo{person}{Francisco Massa}, \bibinfo{person}{Adam Lerer}, \bibinfo{person}{James Bradbury}, \bibinfo{person}{Gregory Chanan}, \bibinfo{person}{Trevor Killeen}, \bibinfo{person}{Zeming Lin}, \bibinfo{person}{Natalia Gimelshein}, \bibinfo{person}{Luca Antiga}, {et~al\mbox{.}}} \bibinfo{year}{2019}\natexlab{}.
\newblock \showarticletitle{{PyTorch: An imperative style, high-performance deep learning library}}.
\newblock \bibinfo{journal}{\emph{Advances in neural information processing systems}}  \bibinfo{volume}{32} (\bibinfo{year}{2019}).
\newblock


\bibitem[Pavlo et~al\mbox{.}(2017)]%
        {pavlo2017self}
\bibfield{author}{\bibinfo{person}{Andrew Pavlo}, \bibinfo{person}{Gustavo Angulo}, \bibinfo{person}{Joy Arulraj}, \bibinfo{person}{Haibin Lin}, \bibinfo{person}{Jiexi Lin}, \bibinfo{person}{Lin Ma}, \bibinfo{person}{Prashanth Menon}, \bibinfo{person}{Todd~C Mowry}, \bibinfo{person}{Matthew Perron}, \bibinfo{person}{Ian Quah}, {et~al\mbox{.}}} \bibinfo{year}{2017}\natexlab{}.
\newblock \showarticletitle{{Self-Driving Database Management Systems.}}. In \bibinfo{booktitle}{\emph{CIDR}}, Vol.~\bibinfo{volume}{4}. \bibinfo{pages}{1}.
\newblock


\bibitem[Pavlo et~al\mbox{.}(2021)]%
        {pavlo2021make}
\bibfield{author}{\bibinfo{person}{Andrew Pavlo}, \bibinfo{person}{Matthew Butrovich}, \bibinfo{person}{Lin Ma}, \bibinfo{person}{Prashanth Menon}, \bibinfo{person}{Wan~Shen Lim}, \bibinfo{person}{Dana Van~Aken}, {and} \bibinfo{person}{William Zhang}.} \bibinfo{year}{2021}\natexlab{}.
\newblock \showarticletitle{Make your database system dream of electric sheep: towards self-driving operation}.
\newblock \bibinfo{journal}{\emph{Proceedings of the VLDB Endowment}} \bibinfo{volume}{14}, \bibinfo{number}{12} (\bibinfo{year}{2021}), \bibinfo{pages}{3211--3221}.
\newblock


\bibitem[Raju et~al\mbox{.}(2017)]%
        {raju2017pebblesdb}
\bibfield{author}{\bibinfo{person}{Pandian Raju}, \bibinfo{person}{Rohan Kadekodi}, \bibinfo{person}{Vijay Chidambaram}, {and} \bibinfo{person}{Ittai Abraham}.} \bibinfo{year}{2017}\natexlab{}.
\newblock \showarticletitle{{Pebblesdb: Building key-value stores using fragmented log-structured merge trees}}. In \bibinfo{booktitle}{\emph{Proceedings of the 26th Symposium on Operating Systems Principles}}. \bibinfo{pages}{497--514}.
\newblock


\bibitem[Ren et~al\mbox{.}(2017)]%
        {ren2017slimdb}
\bibfield{author}{\bibinfo{person}{Kai Ren}, \bibinfo{person}{Qing Zheng}, \bibinfo{person}{Joy Arulraj}, {and} \bibinfo{person}{Garth Gibson}.} \bibinfo{year}{2017}\natexlab{}.
\newblock \showarticletitle{{SlimDB: A space-efficient key-value storage engine for semi-sorted data}}.
\newblock \bibinfo{journal}{\emph{Proceedings of the VLDB Endowment}} \bibinfo{volume}{10}, \bibinfo{number}{13} (\bibinfo{year}{2017}), \bibinfo{pages}{2037--2048}.
\newblock


\bibitem[Sarkar and Athanassoulis(2022)]%
        {sarkar2022dissecting}
\bibfield{author}{\bibinfo{person}{Subhadeep Sarkar} {and} \bibinfo{person}{Manos Athanassoulis}.} \bibinfo{year}{2022}\natexlab{}.
\newblock \showarticletitle{{Dissecting, Designing, and Optimizing LSM-based Data Stores}}. In \bibinfo{booktitle}{\emph{Proceedings of the 2022 International Conference on Management of Data}}. \bibinfo{pages}{2489--2497}.
\newblock


\bibitem[Sarkar et~al\mbox{.}(2023)]%
        {sarkar2023lsm}
\bibfield{author}{\bibinfo{person}{Subhadeep Sarkar}, \bibinfo{person}{Niv Dayan}, {and} \bibinfo{person}{Manos Athanassoulis}.} \bibinfo{year}{2023}\natexlab{}.
\newblock \showarticletitle{The LSM Design Space and its Read Optimizations}. In \bibinfo{booktitle}{\emph{Proceedings of the IEEE International Conference on Data Engineering (ICDE)}}.
\newblock


\bibitem[Sarkar et~al\mbox{.}(2020)]%
        {sarkar2020lethe}
\bibfield{author}{\bibinfo{person}{Subhadeep Sarkar}, \bibinfo{person}{Tarikul~Islam Papon}, \bibinfo{person}{Dimitris Staratzis}, {and} \bibinfo{person}{Manos Athanassoulis}.} \bibinfo{year}{2020}\natexlab{}.
\newblock \showarticletitle{{Lethe: A tunable delete-aware LSM engine}}. In \bibinfo{booktitle}{\emph{Proceedings of the 2020 ACM SIGMOD International Conference on Management of Data}}. \bibinfo{pages}{893--908}.
\newblock


\bibitem[Sarkar et~al\mbox{.}(2022)]%
        {sarkar2022constructing}
\bibfield{author}{\bibinfo{person}{Subhadeep Sarkar}, \bibinfo{person}{Dimitris Staratzis}, \bibinfo{person}{Zichen Zhu}, {and} \bibinfo{person}{Manos Athanassoulis}.} \bibinfo{year}{2022}\natexlab{}.
\newblock \showarticletitle{{Constructing and analyzing the LSM compaction design space}}.
\newblock \bibinfo{journal}{\emph{arXiv preprint arXiv:2202.04522}} (\bibinfo{year}{2022}).
\newblock


\bibitem[Sears and Ramakrishnan(2012)]%
        {sears2012blsm}
\bibfield{author}{\bibinfo{person}{Russell Sears} {and} \bibinfo{person}{Raghu Ramakrishnan}.} \bibinfo{year}{2012}\natexlab{}.
\newblock \showarticletitle{{bLSM: a general purpose log structured merge tree}}. In \bibinfo{booktitle}{\emph{Proceedings of the 2012 ACM SIGMOD International Conference on Management of Data}}. \bibinfo{pages}{217--228}.
\newblock


\bibitem[Shetty et~al\mbox{.}(2013)]%
        {shetty2013building}
\bibfield{author}{\bibinfo{person}{Pradeep~J Shetty}, \bibinfo{person}{Richard~P Spillane}, \bibinfo{person}{Ravikant~R Malpani}, \bibinfo{person}{Binesh Andrews}, \bibinfo{person}{Justin Seyster}, {and} \bibinfo{person}{Erez Zadok}.} \bibinfo{year}{2013}\natexlab{}.
\newblock \showarticletitle{{Building workload-independent storage with VT-trees}}. In \bibinfo{booktitle}{\emph{Presented as part of the 11th $\{$USENIX$\}$ Conference on File and Storage Technologies ($\{$FAST$\}$ 13)}}. \bibinfo{pages}{17--30}.
\newblock


\bibitem[Sun and Li(2018)]%
        {sun13end}
\bibfield{author}{\bibinfo{person}{Ji Sun} {and} \bibinfo{person}{Guoliang Li}.} \bibinfo{year}{2018}\natexlab{}.
\newblock \showarticletitle{{An End-to-End Learning-based Cost Estimator}}.
\newblock \bibinfo{journal}{\emph{Proceedings of the VLDB Endowment}} \bibinfo{volume}{13}, \bibinfo{number}{3} (\bibinfo{year}{2018}).
\newblock


\bibitem[Thonangi and Yang(2017)]%
        {thonangi2017log}
\bibfield{author}{\bibinfo{person}{Risi Thonangi} {and} \bibinfo{person}{Jun Yang}.} \bibinfo{year}{2017}\natexlab{}.
\newblock \showarticletitle{On log-structured merge for solid-state drives}. In \bibinfo{booktitle}{\emph{2017 IEEE 33rd International Conference on Data Engineering (ICDE)}}. IEEE, \bibinfo{pages}{683--694}.
\newblock


\bibitem[Van~Aken et~al\mbox{.}(2017)]%
        {van2017automatic}
\bibfield{author}{\bibinfo{person}{Dana Van~Aken}, \bibinfo{person}{Andrew Pavlo}, \bibinfo{person}{Geoffrey~J Gordon}, {and} \bibinfo{person}{Bohan Zhang}.} \bibinfo{year}{2017}\natexlab{}.
\newblock \showarticletitle{Automatic database management system tuning through large-scale machine learning}. In \bibinfo{booktitle}{\emph{Proceedings of the 2017 ACM international conference on management of data}}. \bibinfo{pages}{1009--1024}.
\newblock


\bibitem[Van~Aken et~al\mbox{.}(2021)]%
        {van2021inquiry}
\bibfield{author}{\bibinfo{person}{Dana Van~Aken}, \bibinfo{person}{Dongsheng Yang}, \bibinfo{person}{Sebastien Brillard}, \bibinfo{person}{Ari Fiorino}, \bibinfo{person}{Bohan Zhang}, \bibinfo{person}{Christian Bilien}, {and} \bibinfo{person}{Andrew Pavlo}.} \bibinfo{year}{2021}\natexlab{}.
\newblock \showarticletitle{An inquiry into machine learning-based automatic configuration tuning services on real-world database management systems}.
\newblock \bibinfo{journal}{\emph{Proceedings of the VLDB Endowment}} \bibinfo{volume}{14}, \bibinfo{number}{7} (\bibinfo{year}{2021}), \bibinfo{pages}{1241--1253}.
\newblock


\bibitem[Vin{\c{c}}on et~al\mbox{.}(2018)]%
        {vinccon2018noftl}
\bibfield{author}{\bibinfo{person}{Tobias Vin{\c{c}}on}, \bibinfo{person}{Sergej Hardock}, \bibinfo{person}{Christian Riegger}, \bibinfo{person}{Julian Oppermann}, \bibinfo{person}{Andreas Koch}, {and} \bibinfo{person}{Ilia Petrov}.} \bibinfo{year}{2018}\natexlab{}.
\newblock \showarticletitle{{Noftl-kv: Tackling write-amplification on kv-stores with native storage management}}. In \bibinfo{booktitle}{\emph{Advances in database technology-EDBT 2018: 21st International Conference on Extending Database Technology, Vienna, Austria, March 26-29, 2018. proceedings}}. University of Konstanz, University Library, \bibinfo{pages}{457--460}.
\newblock


\bibitem[Virtanen et~al\mbox{.}(2020)]%
        {virtanen2020scipy}
\bibfield{author}{\bibinfo{person}{Pauli Virtanen}, \bibinfo{person}{Ralf Gommers}, \bibinfo{person}{Travis~E Oliphant}, \bibinfo{person}{Matt Haberland}, \bibinfo{person}{Tyler Reddy}, \bibinfo{person}{David Cournapeau}, \bibinfo{person}{Evgeni Burovski}, \bibinfo{person}{Pearu Peterson}, \bibinfo{person}{Warren Weckesser}, \bibinfo{person}{Jonathan Bright}, {et~al\mbox{.}}} \bibinfo{year}{2020}\natexlab{}.
\newblock \showarticletitle{{SciPy 1.0: fundamental algorithms for scientific computing in Python}}.
\newblock \bibinfo{journal}{\emph{Nature methods}} \bibinfo{volume}{17}, \bibinfo{number}{3} (\bibinfo{year}{2020}), \bibinfo{pages}{261--272}.
\newblock


\bibitem[Wang et~al\mbox{.}(2014)]%
        {wang2014efficient}
\bibfield{author}{\bibinfo{person}{Peng Wang}, \bibinfo{person}{Guangyu Sun}, \bibinfo{person}{Song Jiang}, \bibinfo{person}{Jian Ouyang}, \bibinfo{person}{Shiding Lin}, \bibinfo{person}{Chen Zhang}, {and} \bibinfo{person}{Jason Cong}.} \bibinfo{year}{2014}\natexlab{}.
\newblock \showarticletitle{{An efficient design and implementation of LSM-tree based key-value store on open-channel SSD}}. In \bibinfo{booktitle}{\emph{Proceedings of the Ninth European Conference on Computer Systems}}. \bibinfo{pages}{1--14}.
\newblock


\bibitem[Wang et~al\mbox{.}(2020)]%
        {wang2020we}
\bibfield{author}{\bibinfo{person}{Xiaoying Wang}, \bibinfo{person}{Changbo Qu}, \bibinfo{person}{Weiyuan Wu}, \bibinfo{person}{Jiannan Wang}, {and} \bibinfo{person}{Qingqing Zhou}.} \bibinfo{year}{2020}\natexlab{}.
\newblock \showarticletitle{Are we ready for learned cardinality estimation?}
\newblock \bibinfo{journal}{\emph{arXiv preprint arXiv:2012.06743}} (\bibinfo{year}{2020}).
\newblock


\bibitem[Wang et~al\mbox{.}(2022)]%
        {wang2022wetune}
\bibfield{author}{\bibinfo{person}{Zhaoguo Wang}, \bibinfo{person}{Zhou Zhou}, \bibinfo{person}{Yicun Yang}, \bibinfo{person}{Haoran Ding}, \bibinfo{person}{Gansen Hu}, \bibinfo{person}{Ding Ding}, \bibinfo{person}{Chuzhe Tang}, \bibinfo{person}{Haibo Chen}, {and} \bibinfo{person}{Jinyang Li}.} \bibinfo{year}{2022}\natexlab{}.
\newblock \showarticletitle{{WeTune: Automatic Discovery and Verification of Query Rewrite Rules}}. In \bibinfo{booktitle}{\emph{Proceedings of the 2022 International Conference on Management of Data}}. \bibinfo{pages}{94--107}.
\newblock


\bibitem[Wu et~al\mbox{.}(2015)]%
        {wu2015lsm}
\bibfield{author}{\bibinfo{person}{Xingbo Wu}, \bibinfo{person}{Yuehai Xu}, \bibinfo{person}{Zili Shao}, {and} \bibinfo{person}{Song Jiang}.} \bibinfo{year}{2015}\natexlab{}.
\newblock \showarticletitle{{LSM-trie: An LSM-tree-based ultra-large key-value store for small data}}. In \bibinfo{booktitle}{\emph{Proceedings of the 2015 USENIX Conference on Usenix Annual Technical Conference}}. USENIX Association, \bibinfo{pages}{71--82}.
\newblock


\bibitem[Yang et~al\mbox{.}(2020)]%
        {yang2020leaper}
\bibfield{author}{\bibinfo{person}{Lei Yang}, \bibinfo{person}{Hong Wu}, \bibinfo{person}{Tieying Zhang}, \bibinfo{person}{Xuntao Cheng}, \bibinfo{person}{Feifei Li}, \bibinfo{person}{Lei Zou}, \bibinfo{person}{Yujie Wang}, \bibinfo{person}{Rongyao Chen}, \bibinfo{person}{Jianying Wang}, {and} \bibinfo{person}{Gui Huang}.} \bibinfo{year}{2020}\natexlab{}.
\newblock \showarticletitle{{Leaper: A learned prefetcher for cache invalidation in LSM-tree based storage engines}}.
\newblock \bibinfo{journal}{\emph{Proceedings of the VLDB Endowment}} \bibinfo{volume}{13}, \bibinfo{number}{12} (\bibinfo{year}{2020}), \bibinfo{pages}{1976--1989}.
\newblock


\bibitem[Ying(2019)]%
        {ying2019overview}
\bibfield{author}{\bibinfo{person}{Xue Ying}.} \bibinfo{year}{2019}\natexlab{}.
\newblock \showarticletitle{An overview of overfitting and its solutions}. In \bibinfo{booktitle}{\emph{Journal of physics: Conference series}}, Vol.~\bibinfo{volume}{1168}. IOP Publishing, \bibinfo{pages}{022022}.
\newblock


\bibitem[Zhang et~al\mbox{.}(2018b)]%
        {zhang2018surf}
\bibfield{author}{\bibinfo{person}{Huanchen Zhang}, \bibinfo{person}{Hyeontaek Lim}, \bibinfo{person}{Viktor Leis}, \bibinfo{person}{David~G Andersen}, \bibinfo{person}{Michael Kaminsky}, \bibinfo{person}{Kimberly Keeton}, {and} \bibinfo{person}{Andrew Pavlo}.} \bibinfo{year}{2018}\natexlab{b}.
\newblock \showarticletitle{{Surf: Practical range query filtering with fast succinct tries}}. In \bibinfo{booktitle}{\emph{Proceedings of the 2018 International Conference on Management of Data}}. \bibinfo{pages}{323--336}.
\newblock


\bibitem[Zhang et~al\mbox{.}(2019)]%
        {zhang2019end}
\bibfield{author}{\bibinfo{person}{Ji Zhang}, \bibinfo{person}{Yu Liu}, \bibinfo{person}{Ke Zhou}, \bibinfo{person}{Guoliang Li}, \bibinfo{person}{Zhili Xiao}, \bibinfo{person}{Bin Cheng}, \bibinfo{person}{Jiashu Xing}, \bibinfo{person}{Yangtao Wang}, \bibinfo{person}{Tianheng Cheng}, \bibinfo{person}{Li Liu}, {et~al\mbox{.}}} \bibinfo{year}{2019}\natexlab{}.
\newblock \showarticletitle{An end-to-end automatic cloud database tuning system using deep reinforcement learning}. In \bibinfo{booktitle}{\emph{Proceedings of the 2019 International Conference on Management of Data}}. \bibinfo{pages}{415--432}.
\newblock


\bibitem[Zhang et~al\mbox{.}(2021)]%
        {zhang2021cdbtune+}
\bibfield{author}{\bibinfo{person}{Ji Zhang}, \bibinfo{person}{Ke Zhou}, \bibinfo{person}{Guoliang Li}, \bibinfo{person}{Yu Liu}, \bibinfo{person}{Ming Xie}, \bibinfo{person}{Bin Cheng}, {and} \bibinfo{person}{Jiashu Xing}.} \bibinfo{year}{2021}\natexlab{}.
\newblock \showarticletitle{{CDBTune+: An efficient deep reinforcement learning-based automatic cloud database tuning system}}.
\newblock \bibinfo{journal}{\emph{The VLDB Journal}} \bibinfo{volume}{30}, \bibinfo{number}{6} (\bibinfo{year}{2021}), \bibinfo{pages}{959--987}.
\newblock


\bibitem[Zhang et~al\mbox{.}(2020)]%
        {zhang2020fpga}
\bibfield{author}{\bibinfo{person}{Teng Zhang}, \bibinfo{person}{Jianying Wang}, \bibinfo{person}{Xuntao Cheng}, \bibinfo{person}{Hao Xu}, \bibinfo{person}{Nanlong Yu}, \bibinfo{person}{Gui Huang}, \bibinfo{person}{Tieying Zhang}, \bibinfo{person}{Dengcheng He}, \bibinfo{person}{Feifei Li}, \bibinfo{person}{Wei Cao}, {et~al\mbox{.}}} \bibinfo{year}{2020}\natexlab{}.
\newblock \showarticletitle{{FPGA-Accelerated Compactions for LSM-based Key-Value Store.}}. In \bibinfo{booktitle}{\emph{FAST}}. \bibinfo{pages}{225--237}.
\newblock


\bibitem[Zhang et~al\mbox{.}(2018a)]%
        {zhang2018elasticbf}
\bibfield{author}{\bibinfo{person}{Yueming Zhang}, \bibinfo{person}{Yongkun Li}, \bibinfo{person}{Fan Guo}, \bibinfo{person}{Cheng Li}, {and} \bibinfo{person}{Yinlong Xu}.} \bibinfo{year}{2018}\natexlab{a}.
\newblock \showarticletitle{{ElasticBF: Fine-grained and Elastic Bloom Filter Towards Efficient Read for LSM-tree-based KV Stores.}}. In \bibinfo{booktitle}{\emph{HotStorage}}.
\newblock


\bibitem[Zhu et~al\mbox{.}(2021)]%
        {zhu2021reducing}
\bibfield{author}{\bibinfo{person}{Zichen Zhu}, \bibinfo{person}{Ju~Hyoung Mun}, \bibinfo{person}{Aneesh Raman}, {and} \bibinfo{person}{Manos Athanassoulis}.} \bibinfo{year}{2021}\natexlab{}.
\newblock \showarticletitle{Reducing bloom filter cpu overhead in lsm-trees on modern storage devices}. In \bibinfo{booktitle}{\emph{Proceedings of the 17th International Workshop on Data Management on New Hardware (DaMoN 2021)}}. \bibinfo{pages}{1--10}.
\newblock


\end{thebibliography}

\end{document}